\documentclass[11pt]{amsart}

\usepackage{amsmath,amssymb,amsthm}
\usepackage{mathtools}
\usepackage{enumerate}
\usepackage[hidelinks]{hyperref}
\usepackage{cleveref}
\usepackage{tikz}
\usetikzlibrary{positioning}
\usepackage{algorithm}
\usepackage{algpseudocode}
\usepackage{booktabs}

\newtheorem{theorem}{Theorem}[section]
\newtheorem{lemma}[theorem]{Lemma}
\newtheorem{proposition}[theorem]{Proposition}
\newtheorem{corollary}[theorem]{Corollary}

\theoremstyle{definition}
\newtheorem{definition}[theorem]{Definition}
\newtheorem{example}[theorem]{Example}
\newtheorem{remark}[theorem]{Remark}

\newcommand{\ZZ}{\mathbb{Z}}
\newcommand{\QQ}{\mathbb{Q}}
\newcommand{\RR}{\mathbb{R}}

\newcommand{\KK}{\mathbb{K}}
\newcommand{\PM}{\mathrm{PM}}
\newcommand{\sgn}{\operatorname{sgn}}
\newcommand{\ord}{\operatorname{ord}}
\newcommand{\col}{\operatorname{col}}

\DeclareMathOperator{\rk}{rk}

\title[Bipartite Exact Matching in P]{Bipartite Exact Matching in $\mathrm{P}$}
\author{Yuefeng Du\footnote{Email: \texttt{yf.du@cityu.edu.hk}}\\
{\upshape\mdseries Department of Computer Science\\ City University of Hong Kong}}
\date{}
\begin{document}

\begin{abstract}
The \emph{Exact Matching} problem asks whether a bipartite graph with edges colored red and blue admits a perfect matching with exactly~$t$ red edges.
Introduced by Papadimitriou and Yannakakis in 1982, the problem has resisted deterministic polynomial-time algorithms for over four decades, despite admitting a randomized solution via the Schwartz--Zippel lemma since 1987.

We establish the Affine-Slice Nonvanishing Theorem (ASNC) for all bipartite braces: a Vandermonde-weighted determinant polynomial is nonzero whenever the exact-$t$ fiber is nonempty.
This yields a deterministic $O(n^6)$ algorithm for Exact Matching on all bipartite graphs via the tight-cut decomposition into brace blocks.
The proof proceeds by structural induction on McCuaig's brace decomposition.
We handle the McCuaig exceptional families via a parity-resolved cylindric-network positivity argument, the replacement determinant algebra, and the narrow-extension cases (KA, $J3 \to D1$).
For the superfluous-edge step, we introduce two closure tools: a matching-induced Two-extra Hall theorem that resolves the rank-$(m{-}2)$ branch via projective-collapse contradiction, and a distinguished-state $q$-circuit lemma that eliminates the rank-$(m{-}1)$ branch entirely by showing that any minimal dependent set containing the superfluous state forces rank $m{-}2$.
A Lean~4 formalization accompanies the paper.
The formalization reduces the main theorem to eight explicit hypotheses corresponding to results proved here and in McCuaig~(2001), with all algebraic tools, the induction skeleton, and the combinatorial infrastructure fully machine-checked.
\end{abstract}

\maketitle

\tableofcontents

\section{Introduction}\label{sec:intro}

\subsection{The Exact Matching problem}

Let $G = (A \cup B, E)$ be a bipartite graph with $|A| = |B| = n$, equipped with an edge coloring $\rho \colon E \to \{0,1\}$ (blue/red) and a target integer~$t$.
The \emph{Exact Matching} (EM) problem asks:

\begin{quote}
Does $G$ admit a perfect matching $M$ with exactly $t$ red edges, i.e., $|\{e \in M : \rho(e) = 1\}| = t$?
\end{quote}

This problem was introduced by Papadimitriou and Yannakakis~\cite{PY82} in 1982, who conceived it while studying constrained spanning tree problems.
At the time of its introduction, the problem was conjectured to be NP-hard.
A few years later, Mulmuley, Vazirani, and Vazirani~\cite{MVV87} gave a randomized polynomial-time algorithm using the Schwartz--Zippel lemma~\cite{SZ80,Zippel79} and a weight isolation technique.
This placed Exact Matching in the complexity class~$\mathrm{RP}$, making it one of the few natural problems known to separate $\mathrm{P}$ from~$\mathrm{RP}$ conditionally.

Despite sustained effort over four decades, no deterministic polynomial-time algorithm has been found for the general case.
Deterministic progress has been confined to special graph classes and parameterized regimes: planar graphs~\cite{Yuster2012}, exact weight matching in bipartite graphs with bounded neighborhood diversity or bandwidth~\cite{ElMaalouly2025}, correct-parity matching~\cite{ElMaalouly2023}, and dense graphs~\cite{ElMaalouly2024STACS}.
On the algebraic side, Fenner, Gurjar, and Thierauf~\cite{Fenner2003} placed bipartite perfect matching in quasi-NC.
Gurjar and Thierauf~\cite{GurjarThierauf2020} gave a deterministic parallel algorithm, and Svensson and Tarnawski~\cite{Svensson2017} extended the quasi-NC result to general graphs.
Karzanov~\cite{Karzanov87} studied exact weight matching in complete bipartite graphs.
In the online setting, bipartite matching admits an optimal competitive ratio of $1 - 1/e$~\cite{KVV90}.
The problem remains a natural candidate for testing the $\mathrm{P} = \mathrm{RP}$ hypothesis and sits at the intersection of algebra, combinatorics, and complexity theory.

\subsection{Our result}

\begin{theorem}[Main Theorem]\label{thm:main}
There exists a deterministic algorithm that solves the Exact Matching problem on bipartite graphs in $O(n^3)$ determinant evaluations, giving $O(n^6)$ arithmetic complexity.
\end{theorem}

\noindent
The algorithm is based on a new algebraic identity.
Define the \emph{Vandermonde-weighted matrix}
\[
M_G(x,\lambda)_{ij} =
\begin{cases}
x^{\rho(i,j)} (\lambda + i)^j, & \text{if } (i,j) \in E, \\
0, & \text{otherwise},
\end{cases}
\]
where rows are indexed by $A = \{0, \ldots, n-1\}$ and columns by $B = \{0, \ldots, n-1\}$.
Let $D_G(x,\lambda) = \det M_G(x,\lambda)$ and define the \emph{exact-$t$ Vandermonde polynomial}
\[
P_t(\lambda) = [x^t]\, D_G(x,\lambda) = \sum_{\substack{\sigma \in \PM(G) \\ |\{i : \rho(i,\sigma(i))=1\}| = t}} \sgn(\sigma) \prod_{i=0}^{n-1} (\lambda+i)^{\sigma(i)}.
\]

The exact-$t$ fiber $F_{G,t} = \{\sigma \in \PM(G) : \text{red count} = t\}$ satisfies: $F_{G,t} = \varnothing$ implies $P_t \equiv 0$ (immediate from the Leibniz expansion). The converse, that a nonempty fiber forces $P_t \not\equiv 0$, is the content of:

\begin{theorem}[Affine-Slice Nonvanishing, ASNC]\label{conj:asnc}
For every bipartite brace $G$ on $n+n$ vertices, every edge coloring $\rho$, and every target~$t$, if the exact-$t$ fiber is nonempty, then $P_t(\lambda) \not\equiv 0$.
\end{theorem}

\begin{remark}
The proof handles McCuaig's exceptional families, narrow extensions (KA, $J3 \to D1$), and both rank branches of the superfluous-edge step: the rank-$(m{-}2)$ branch via the Two-extra Hall closure, and the rank-$(m{-}1)$ branch via the $q$-circuit lemma (\Cref{prop:q-circuit-rankm1-impossible}).
\end{remark}

\begin{remark}[ASNC scope]\label{rem:decision-reduction}
The $O(n^6)$ algorithm solves Exact Matching for all bipartite graphs, brace or not.
ASNC itself is stated and proved only for braces.
The gap is bridged by block-independent testing: the tight-cut decomposition splits any matching-covered bipartite graph into brace blocks, each tested separately via ASNC.
The top-level answer is recovered by a knapsack DP over target partitions.
ASNC is not claimed for non-brace graphs, since the convolution $P_{G,t} = \sum_{t_1+t_2=t} P_{G_1,t_1} \cdot P_{G_2,t_2}$ can cancel even when individual factors are nonzero.
See Step~1 of the proof of \Cref{thm:reduction} for details.
\end{remark}

\noindent
The proof proceeds by structural induction on brace graphs, with the key closure step being a matching-induced two-extra Hall theorem (\Cref{thm:two-extra-hall-closure}).

\Cref{thm:main} follows from \Cref{conj:asnc} via the decomposition-based algorithm of \Cref{prop:asnc-implies-em}: compute the tight-cut decomposition into brace blocks, test each block for ASNC nonvanishing, and reconstruct the answer by dynamic programming over target partitions.
Since ASNC is now proved for all braces (\Cref{conj:asnc}), the algorithm is unconditional.

\subsection{Proof overview}

The proof of ASNC (\Cref{conj:asnc}) proceeds by structural induction on bipartite matching-covered graphs, following the brace-theoretic decomposition framework of McCuaig~\cite{McCuaig2001}.

\paragraph{Step 1: Tight-cut reduction.}
Using Hall-block convolution (\Cref{thm:hall-block}), we reduce to brace graphs (bipartite graphs satisfying the Hall$+2$ condition: $|N(S)| \geq |S| + 2$ for all proper nonempty $S$).
This step is an exact determinant factorization that holds for all~$n$.

\paragraph{Step 2: Minimal vs.\ nonminimal braces.}
For \emph{minimal braces} (those where deleting any edge destroys the brace property), McCuaig's theory provides a classification into narrow local extension types (stable extension, index-1, index-2) and three exceptional families (biwheels, prisms, M\"obius ladders).
We handle:
\begin{itemize}
\item Stable extension / index-2: via \emph{two-node interpolation} (KA, \Cref{cor:ka}), using the coprimality of distinct Vandermonde bases.
\item Index-1: via \emph{three-node Vandermonde elimination} (\Cref{cor:j3-to-d1}), reducing to same-base divisibility.
\item McCuaig's exceptional families: biwheels via jet extraction at a unique-minimizer row, prisms and M\"obius ladders via a parity-resolved cylindric-network positivity argument (\Cref{sec:mccuaig}).
\end{itemize}

For \emph{nonminimal braces}, the key step is the superfluous-edge argument.
A superfluous edge $e = (r,c)$ in a brace can be deleted while preserving the brace property.
The exact identity
\[
P_t(G) = P_t(H) + \varepsilon (\lambda+r)^c P^{r,c}_{H, t-\rho(e)}
\]
where $H = G \setminus e$ relates the polynomials of $G$ and~$H$.

\paragraph{Step 3: The superfluous-edge step (\Cref{thm:p84}).}
This is the heart of the paper.
Assuming $P_t(G) \equiv 0$ for contradiction with $F_{G,t} \neq \varnothing$, a packet rank dichotomy arises (\Cref{prop:p92}): the lower packet $V_T$ on a minimal dependent set $T$ containing the superfluous state has rank $m{-}1$ or $m{-}2$.
Two tools resolve both branches:
\begin{enumerate}[(i)]
\item The \emph{$q$-circuit lemma} (\Cref{prop:q-circuit-rankm1-impossible}) eliminates the rank-$(m{-}1)$ branch entirely.
If $\rk(V_T) = m{-}1$, then Hall's condition holds for all subsets of $T$ in the support graph of the augmented packet, forcing $\rk(B_T) = m$ by the masked minor criterion, contradicting minimal dependence ($\rk(B_T) = m{-}1$).
\item The \emph{Two-extra Hall closure} (\Cref{thm:two-extra-hall-closure}) handles the rank-$(m{-}2)$ branch.
The exact-$t$ perfect matching of $G$ using $e$ induces a support-graph matching for the augmented packet, yielding a nonzero pair minor $M_{w_c^G, w_d} \neq 0$ for each upper active column $d \in D^+$.
This contradicts the projective collapse required by rank $m{-}2$.
\end{enumerate}

\subsection{Proof architecture}

The proof combines three ingredients that have not previously appeared together:
\begin{enumerate}
\item \emph{Graph-structural induction} on McCuaig's brace decomposition tree, proving ASNC for all braces via local moves (superfluous-edge deletion, narrow extensions). The tight-cut decision reduction then extends the algorithm to all bipartite graphs.
\item \emph{Algebraic elimination} via replacement determinants and generalized Vandermonde coprimality, converting local graph moves into polynomial divisibility contradictions.
\item \emph{Matching-induced two-extra Hall closure}, exploiting the perfect matching in the support graph of the actual $G$-packet to exclude the residual projective-collapse case.
\end{enumerate}

\begin{figure}[ht]
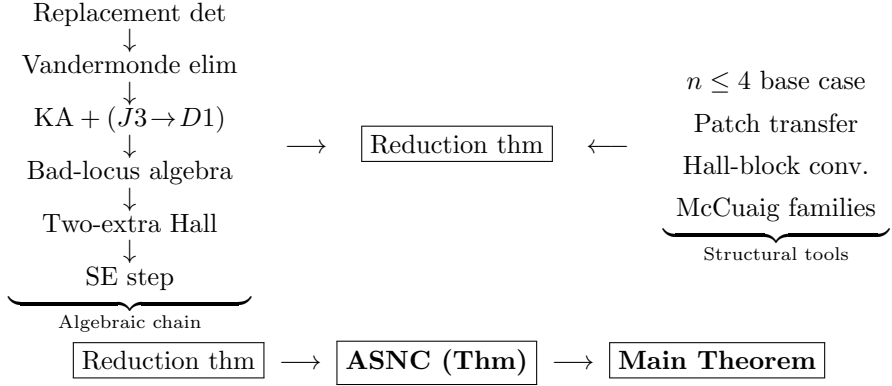

\centering
\small
\[
\underbrace{%
\begin{array}{c}
\text{Replacement det} \\[-2pt] \downarrow \\[-2pt]
\text{Vandermonde elim} \\[-2pt] \downarrow \\[-2pt]
\text{KA} + (J3 \!\to\! D1) \\[-2pt] \downarrow \\[-2pt]
\text{Bad-locus algebra} \\[-2pt] \downarrow \\[-2pt]
\text{Two-extra Hall} \\[-2pt] \downarrow \\[-2pt]
\text{SE step}
\end{array}}_{\text{Algebraic chain}}
\quad\longrightarrow\quad
\boxed{\text{Reduction thm}}
\quad\longleftarrow\quad
\underbrace{%
\begin{array}{c}
n \leq 4 \text{ base case} \\[4pt]
\text{Patch transfer} \\[4pt]
\text{Hall-block conv.} \\[4pt]
\text{McCuaig families}
\end{array}}_{\text{Structural tools}}
\]
\vspace{-4pt}
\[
\boxed{\text{Reduction thm}} \;\longrightarrow\; \boxed{\textbf{ASNC (Thm)}} \;\longrightarrow\; \boxed{\textbf{Main Theorem}}
\]
\caption{Proof dependency structure. The algebraic chain (left) and structural tools (right) both feed into the reduction theorem, yielding the main result.}
\label{fig:dependency}
\end{figure}

The base case ($n \leq 4$) is proved analytically (\Cref{thm:universal-small}).
The $n = 5$ exhaustive verification ($517$M checks, zero failures) provides independent computational evidence but is not part of the induction chain.
A Lean~4 formalization accompanies the paper (\Cref{app:lean}).
The formalization reduces the main theorem to eight explicit hypotheses corresponding to results proved here and in McCuaig~\cite{McCuaig2001}, with all algebraic tools, the induction skeleton, and the combinatorial infrastructure fully machine-checked.
See \Cref{app:lean} for the precise frontier.
The McCuaig exceptional families (biwheels, prisms, M\"obius ladders) are handled in \Cref{sec:mccuaig}: the biwheel case uses row-initial-form uniqueness, while prisms and M\"obius ladders reduce to a parity-resolved cylindric-network positivity argument via width-$2$ path and cyclic branch factorization.

\subsection{Organization}

\Cref{sec:prelim} establishes notation and the Vandermonde specialization.
\Cref{sec:framework} presents the structural induction framework based on brace theory.
\Cref{sec:replacement} develops the replacement determinant algebra.
\Cref{sec:mccuaig} handles the McCuaig exceptional families.
\Cref{sec:badlocus} establishes the bad-locus characterizations for the four operator shapes.
\Cref{sec:superfluous} proves the superfluous-edge step, completing the induction.
\Cref{sec:algorithm} describes the resulting algorithm.

\section{Preliminaries}\label{sec:prelim}

\subsection{Bipartite graphs and perfect matchings}

Throughout, $G = (A \cup B, E)$ denotes a bipartite graph with parts $A = \{0, \ldots, n-1\}$ (rows) and $B = \{0, \ldots, n-1\}$ (columns), equipped with an edge coloring $\rho \colon E \to \{0,1\}$.
A \emph{perfect matching} of~$G$ is a bijection $\sigma \colon A \to B$ such that $(i, \sigma(i)) \in E$ for all~$i$.
The existence of a perfect matching can be characterized by Tutte's theorem~\cite{Tutte47} or, in the bipartite case, by Hall's marriage theorem~\cite{Hall35}.
We write $\PM(G)$ for the set of perfect matchings of~$G$.
The matching polytope of~$G$ (the convex hull of characteristic vectors of perfect matchings) is a face of the Birkhoff polytope~\cite{Birkhoff46,Edmonds65}.

For a target $t \in \{0, \ldots, n\}$, the \emph{exact-$t$ fiber} is
\[
F_{G,t} = \Bigl\{\sigma \in \PM(G) : \sum_{i=0}^{n-1} \rho(i, \sigma(i)) = t \Bigr\}.
\]

\subsection{The Vandermonde specialization}

\noindent
The central algebraic object is the Vandermonde-weighted matrix, whose determinant separates exact-count matching fibers.

\begin{definition}[Vandermonde-weighted matrix]\label{def:vandermonde-matrix}
For a colored bipartite graph $(G, \rho)$, define the matrix $M_G(x, \lambda) \in \QQ[x,\lambda]^{n \times n}$ by
\[
M_G(x,\lambda)_{ij} =
\begin{cases}
x^{\rho(i,j)} (\lambda + i)^j, & \text{if } (i,j) \in E(G), \\
0, & \text{otherwise}.
\end{cases}
\]
\end{definition}

\begin{definition}[Exact-$t$ Vandermonde polynomial]\label{def:Pt}
The \emph{generating determinant} is $D_G(x,\lambda) = \det M_G(x,\lambda)$, and the \emph{exact-$t$ polynomial} is
\[
P_t(\lambda) = P_{G,t}(\lambda) = [x^t]\, D_G(x,\lambda).
\]
\end{definition}

Expanding the determinant:
\[
P_t(\lambda) = \sum_{\sigma \in F_{G,t}} \sgn(\sigma) \prod_{i=0}^{n-1} (\lambda+i)^{\sigma(i)}.
\]
Each term is a product of $n$ powers of distinct linear forms $(\lambda+0), (\lambda+1), \ldots, (\lambda+n-1)$ with exponents forming a permutation of a subset of $\{0, \ldots, n-1\}$.

\begin{proposition}[ASNC gives deterministic EM]\label{prop:asnc-implies-em}
By \Cref{conj:asnc}, Exact Matching on bipartite graphs with $n$ vertices per side can be solved deterministically in $O(n^6)$ arithmetic operations.
\end{proposition}

\begin{proof}
The algorithm proceeds in three steps.

\smallskip
\noindent\textbf{Step 1: Tight-cut decomposition.}
If $G$ has no perfect matching, return \textsc{No}.
Otherwise, compute the tight-cut decomposition of~$G$ into brace blocks $G_1, \ldots, G_s$ in $O(n^3)$ time~\cite{CLM2002}.
Every matching-covered bipartite graph with no tight cut is a brace~\cite{LovaszPlummer86}.

\smallskip
\noindent\textbf{Step 2: Per-block ASNC test.}
For each block $G_j$ with $n_j$ vertices per side and each target $t' \in \{0, \ldots, n_j\}$,
test whether $P_{t'}(G_j, \lambda) \not\equiv 0$.
The polynomial $P_{t'}(G_j, \lambda)$ has degree at most $\binom{n_j}{2}$, so evaluating at $\binom{n_j}{2}+1$ values of~$\lambda$ suffices.
For each $\lambda$-value, evaluate $D_{G_j}(x,\lambda)$ at $n_j+1$ values of~$x$ using Bareiss ($O(n_j^3)$ per evaluation), then recover $P_{t'}(\lambda)$ via Lagrange interpolation in~$x$.
By \Cref{conj:asnc} (proved for all braces), $F_{G_j,t'} \neq \varnothing$ if and only if $P_{t'}(G_j) \not\equiv 0$.
The cost per block is $O(n_j^6)$, and since $\sum_j n_j \le n$, the total is $O(n^6)$.

\smallskip
\noindent\textbf{Step 3: Target-partition DP.}
Since matchings decompose along tight cuts, $F_{G,t} \neq \varnothing$ iff there exist $t_1 + \cdots + t_s = t$ with $F_{G_j, t_j} \neq \varnothing$ for each~$j$.
A standard knapsack-style DP over the blocks determines this in $O(n^2)$.

\smallskip
\noindent
Note that ASNC is applied only to brace blocks: the convolution $P_{G,t} = \sum_{t_1+t_2=t} P_{G_1,t_1} \cdot P_{G_2,t_2}$ for a non-brace graph could cancel even when individual factors are nonzero, so the algorithm tests blocks independently rather than evaluating $P_t(G)$ directly.
\end{proof}

\subsection{Boundary-minor states}

\begin{definition}[Boundary-minor polynomial]\label{def:boundary-minor}
For paired subsets $U \subseteq A$, $V \subseteq B$ with $|U| = |V|$, the \emph{$(U,V)$-boundary-minor polynomial} is
\[
P^{U,V}_{G,s}(\lambda) = [x^s]\, \det M_G(x,\lambda)[A \setminus U,\; B \setminus V].
\]
The \emph{full exact state} of~$G$ is the collection $\Omega_G = (P^{U,V}_{G,s})_{(U,V,s)}$.
\end{definition}

\noindent
The visible output $P_t = P^{\varnothing,\varnothing}_{G,t}$ is the top coordinate of this state.
The boundary-minor state is the correct induction-closed object: deleting a single edge shears the visible coordinate by a one-hole coordinate, so arguments on the visible output alone cannot recurse.


\subsection{Foundational properties}\label{sec:foundational}

This subsection collects four structural results that underpin the induction.
The first (cycle reachability) characterizes the combinatorial structure of achievable red counts. The second (affine closure) establishes the polyhedral geometry of exact-$t$ fibers. The third (universal nonvanishing at small~$n$) provides the base case for the structural induction, and the fourth (row-initial-form nonvanishing) gives a sufficient condition used in the biwheel proof.


\begin{theorem}[Cycle reachability]\label{thm:cycle-reach}
Let $G = (A \cup B, E)$ be a bipartite graph with edge coloring $\rho$ and let $M_0$ be any fixed perfect matching of~$G$.
For each even-length alternating cycle $C$ (edges alternating between $M_0$ and $E \setminus M_0$), define the \emph{red-count displacement}
\[
\Delta(C) = |\{e \in C \setminus M_0 : \rho(e) = 1\}| - |\{e \in C \cap M_0 : \rho(e) = 1\}|.
\]
Then the set of red counts achievable by perfect matchings of~$G$ is exactly
\[
\Bigl\{\,\bigl|\,M_0 \cap R\,\bigr| + \textstyle\sum_i \Delta(C_i) \;:\;
\{C_i\} \text{ vertex-disjoint alternating cycles}\Bigr\},
\]
where $R = \rho^{-1}(1)$ denotes the red edge set.
\end{theorem}

\begin{proof}
The proof rests on the classical symmetric-difference decomposition for matchings.

\noindent\textbf{Forward direction.}
Let $M$ be any perfect matching.
The symmetric difference $M_0 \oplus M = (M_0 \setminus M) \cup (M \setminus M_0)$ decomposes into vertex-disjoint even-length alternating cycles $C_1, \ldots, C_k$ (every vertex has degree exactly~$0$ or~$2$ in $M_0 \oplus M$, and bipartiteness forces even length).
For each $C_i$, the matching $M$ uses the non-$M_0$ edges of~$C_i$ and agrees with $M_0$ outside $\bigcup_i C_i$.
Therefore
\[
|M \cap R| = |M_0 \cap R| + \sum_{i=1}^k \Delta(C_i).
\]

\noindent\textbf{Reverse direction.}
Conversely, given any collection $\{C_i\}$ of vertex-disjoint alternating cycles, flipping $M_0$ along $\bigcup_i C_i$ (replacing $M_0 \cap C_i$ by $C_i \setminus M_0$ in each cycle) yields a valid perfect matching~$M'$.
The matching property is preserved because within each cycle the substitution replaces a perfect matching of the cycle's vertex set with another, and outside the cycles nothing changes.
The red count of~$M'$ is $|M_0 \cap R| + \sum_i \Delta(C_i)$ by construction.
\end{proof}

\begin{remark}\label{rem:cycle-reach-general}
\Cref{thm:cycle-reach} holds for general (non-bipartite) graphs as well, provided one restricts to \emph{even-length} alternating cycles.
Odd-length alternating walks (involving blossom structure) do not give valid perfect-matching transitions: flipping along an odd walk disrupts the matching constraint at the blossom's stem vertex.
\end{remark}

\noindent
Cycle reachability reduces Exact Matching to a constrained cycle-packing problem: given $M_0$, one seeks a collection of vertex-disjoint alternating cycles whose displacements sum to $t - |M_0 \cap R|$.
The difficulty lies in the interaction between vertex-disjointness and displacement values: the set of achievable displacement sums is not, in general, a numerical semigroup or even a contiguous interval.
For instance, $K_{n,n}$ with $R = M_0$ (the identity matching) achieves every red count except $t = n-1$, since a permutation with $n-1$ fixed points necessarily has~$n$ fixed points.


\begin{theorem}[Affine closure of exact-$t$ fibers]\label{thm:affine-closure}
Let $G = (A \cup B, E)$ be a bipartite graph with coloring~$\rho$ and target~$t$.
The exact-$t$ fiber $F_{G,t}$ is the intersection of a Birkhoff face with an affine hyperplane.
Specifically:
\begin{enumerate}[(i)]
\item The \emph{support face} $\mathcal{B}_G = \operatorname{conv}\{P_\sigma : \sigma \in \PM(G)\}$ is a face of the Birkhoff polytope $B_n = \{X \in \RR_{\geq 0}^{n \times n} : X\mathbf{1} = \mathbf{1},\, X^\top \mathbf{1} = \mathbf{1}\}$, determined by the zero constraints $X_{ij} = 0$ for $(i,j) \notin E(G)$.
\item The coloring~$\rho$ defines the weight matrix $W \in \{0,1\}^{n \times n}$ with $W_{ij} = \rho(i,j)$ (and $W_{ij} = 0$ for $(i,j) \notin E$).
The red-count constraint is the affine hyperplane $\mathcal{H}_t = \{X : \langle W, X \rangle = t\}$, where $\langle W, X \rangle = \sum_{i,j} W_{ij} X_{ij}$.
\item The fiber satisfies
\[
F_{G,t} = \PM(G) \cap \{\sigma : \langle W, P_\sigma \rangle = t\} = \operatorname{vert}(\mathcal{B}_G \cap \mathcal{H}_t).
\]
In particular, $F_{G,t}$ is the vertex set of a polytope $\mathcal{B}_G \cap \mathcal{H}_t$ that is itself a face of a Birkhoff face intersected with a hyperplane.
\end{enumerate}
\end{theorem}

\begin{proof}
Every perfect matching $\sigma \in \PM(G)$ corresponds to a permutation matrix $P_\sigma \in B_n$ with $P_\sigma(i,j) = 1$ iff $j = \sigma(i)$.
The Birkhoff--von Neumann theorem identifies $B_n$ as the convex hull of all $n!$ permutation matrices.
The support constraint $(i,j) \notin E(G) \Rightarrow X_{ij} = 0$ defines a face $\mathcal{B}_G$ of $B_n$ whose vertices are exactly $\{P_\sigma : \sigma \in \PM(G)\}$.

The red count of $\sigma$ is $\sum_i W_{i,\sigma(i)} = \langle W, P_\sigma \rangle$.
So $\sigma \in F_{G,t}$ if and only if $P_\sigma$ lies in the hyperplane $\mathcal{H}_t$.
The set $\mathcal{B}_G \cap \mathcal{H}_t$ is a polytope (intersection of a polytope with a hyperplane), and its vertices are a subset of the vertices of $\mathcal{B}_G$, namely those satisfying $\langle W, P_\sigma \rangle = t$.
These are precisely the elements of $F_{G,t}$.
\end{proof}

\begin{corollary}[Parallelogram completion]\label{cor:parallelogram}
If $\sigma_1, \sigma_2, \sigma_3 \in F_{G,t}$ and $P_{\sigma_4} = P_{\sigma_1} + P_{\sigma_3} - P_{\sigma_2}$ happens to be a permutation matrix in $\mathcal{B}_G$, then $\sigma_4 \in F_{G,t}$.
\end{corollary}

\begin{proof}
The affine combination preserves the hyperplane constraint:
$\langle W, P_{\sigma_4} \rangle
= \langle W, P_{\sigma_1} \rangle + \langle W, P_{\sigma_3} \rangle
  - \langle W, P_{\sigma_2} \rangle = t$.
\end{proof}

\begin{remark}\label{rem:affine-closure}
This closure under affine combinations within the Birkhoff face is the key distinction between exact-count fibers and arbitrary subsets of~$\PM(G)$: the former are ``rigid'' in a polyhedral sense.
Arbitrary subsets of $S_n$ \emph{can} produce $P_S(\lambda) \equiv 0$ (as happens for certain unions $H \sqcup H^\uparrow$ at $n = 5$), but the affine rigidity of exact-$t$ fibers prevents the signed cancellation needed for vanishing.
Every proper affinely closed subset of the permutation configuration is equivalent to a full equality layer $S = \{\sigma \in S_n : c_W(\sigma) = t\}$, so ASNC is equivalent to the statement that no nonempty proper affinely closed subset of $S_n$ has $P_S(\lambda) \equiv 0$.
\end{remark}


\begin{theorem}[Universal nonvanishing for $n \leq 4$]\label{thm:universal-small}
For $n \leq 4$ and \emph{every} nonempty subset $S \subseteq S_n$ (not necessarily an exact-$t$ fiber), the polynomial
\[
P_S(\lambda) = \sum_{\sigma \in S} \sgn(\sigma) \prod_{i=0}^{n-1} (\lambda+i)^{\sigma(i)}
\]
is not identically zero.
\end{theorem}

\noindent
This result is strictly stronger than ASNC for $n \leq 4$: it asserts nonvanishing for \emph{all} nonempty permutation families, without any affine-closure or graph-realizability hypothesis.
The first pathological vanishing for unrestricted subsets occurs at $n = 5$, where a specific 24-element subset $S \subset S_5$ satisfies $P_S(\lambda) \equiv 0$.
However, this subset is not affinely closed, and its affine closure (96 elements) has a nonzero polynomial.

\begin{proof}
The proof proceeds by complete enumeration at small $n$, combined with an inductive Newton-chain argument.

\noindent\textbf{Case $n = 2$.}
$S_2 = \{\mathrm{id}, (01)\}$ has three nonempty subsets.
Direct computation gives $P_{\{\mathrm{id}\}}(\lambda) = \lambda^0 (\lambda+1)^1 = \lambda + 1$, $P_{\{(01)\}}(\lambda) = -\lambda^1 (\lambda+1)^0 = -\lambda$, and $P_{S_2}(\lambda) = (\lambda+1) - \lambda = 1$.
All three are nonzero.

\noindent\textbf{Case $n = 3$.}
Apply the left-boundary recurrence: for any nonempty $S \subseteq S_3$, expand along row~$0$.
Each branch $S_a = \{\sigma \in S : \sigma(0) = a\}$ yields a sub-family in $S_2$.
The boundary-gap vanishing lemma shows that vanishing of $P_S$ would force the leading branch polynomial to vanish at a specific point, contradicting the $n = 2$ nonvanishing.

\noindent\textbf{Case $n = 4$.}
This is the substantial case.
Apply the Newton chain identity: if $P_S(\lambda) \equiv 0$, expand along row~$0$ with $z = \lambda$, giving
\[
\sum_{a \in \{0,1,2,3\}} (-1)^a z^a Q_a(z+1) = 0,
\]
where $Q_a(\mu) = P_{L_a(S)}(\mu)$ is the branch polynomial for permutations with $\sigma(0) = a$, defined on subsets of $S_3$ with bases shifted by~$1$.

The Newton chain is a triangular system in the jets of $Q_a$ at $\mu = 1$.
Exhaustive catalog of all $64 = 2^6$ nonempty subsets of $S_3$ (for each of four exponent sets arising from column deletion) shows that the zeroth jet $j_0$ is forced to vanish, which uniquely determines $Q_0$ and $Q_1$.
The remaining pair equation $Q_2 - (z-1)Q_3 = -z(z+2)^2$ has no solution among $64^2$ candidate pairs.
This contradiction holds for every nonempty $S \subseteq S_4$.
\end{proof}

\begin{remark}\label{rem:p23-base-case}
\Cref{thm:universal-small} serves as the base case for the structural induction in two senses.
First, it covers all bipartite graphs on at most $4+4$ vertices without any structural hypothesis.
Second, the McCuaig exceptional families require direct verification: the biwheel case uses a jet extraction at a unique-minimizer row (reducing to $n \leq 4$ where \Cref{thm:universal-small} applies), while prisms and M\"obius ladders are closed by a parity-resolved cylindric-network argument that does not reduce to small~$n$.
The fact that universal nonvanishing holds at $n \leq 4$ (without needing the affine-closure hypothesis) is essential for the biwheel case, since the jet-operator extraction does not preserve affine closure, only the permutation-family structure.
\end{remark}


\begin{definition}[Row Laplace expansion]\label{def:row-laplace}
For a nonempty family $F \subseteq \PM(G)$ and a row $r \in A$, define the \emph{$r$-branch at column~$a$} as the subfamily $F_a^{(r)} = \{\sigma \in F : \sigma(r) = a\}$.
The Laplace expansion of $P_F(\lambda)$ along row~$r$ is
\[
P_F(\lambda) = \sum_{\substack{a \in B \\ F_a^{(r)} \neq \varnothing}} (-1)^{r+a}\, (\lambda+r)^a \cdot P_{F_a^{(r)} \setminus \{r\}}(\lambda),
\]
where $P_{F_a^{(r)} \setminus \{r\}}$ denotes the polynomial of the $(n-1) \times (n-1)$ sub-family obtained by deleting row~$r$ and column~$a$.
\end{definition}

\begin{definition}[Row initial form]\label{def:row-initial}
For row~$r$, let $\tau_r(a) = a$ be the exponent contributed by column~$a$ in the Vandermonde weighting (recall the entry $(\lambda+r)^a$ at position $(r,a)$).
Define the \emph{minimum exponent}
\[
m_r = \min\{a : F_a^{(r)} \neq \varnothing\},
\]
the \emph{tie set} $S_r^{\min} = \{a : a = m_r,\, F_a^{(r)} \neq \varnothing\}$ (when the Vandermonde specialization assigns equal exponents to multiple columns used by row~$r$), and the \emph{row-$r$ initial form}
\[
I_r(\lambda) = \sum_{a \in S_r^{\min}} (-1)^{r+a} (\lambda+r)^a \cdot \operatorname{in}_0\bigl(P_{F_a^{(r)} \setminus \{r\}}(\lambda)\bigr),
\]
where $\operatorname{in}_0(f)$ denotes the leading term of $f$ in the $(\lambda+r)$-adic expansion.
When $|S_r^{\min}| = 1$ (unique minimizer), the tie set is a singleton $\{a_0\}$, and $I_r$ simplifies to a single nonzero term.
\end{definition}

\begin{theorem}[Row-initial-form nonvanishing]\label{thm:row-initial}
Let $F \subseteq \PM(G)$ be a nonempty family.
If there exists a row~$r$ such that the row-$r$ initial form $I_r(\lambda)$ is not identically zero, then $P_F(\lambda) \not\equiv 0$.
\end{theorem}

\begin{proof}
In the Laplace expansion of $P_F$ along row~$r$, group terms by their $(\lambda+r)$-adic order.
Write $\mu_r = \lambda + r$ and expand each branch polynomial as $P_{F_a^{(r)} \setminus \{r\}}(\lambda) = \sum_{k \geq 0} c_{a,k}\, \mu_r^k$.
Then
\[
P_F(\lambda) = \sum_{a} (-1)^{r+a} \mu_r^a \sum_{k} c_{a,k}\, \mu_r^k = \sum_{j \geq 0} \mu_r^j \biggl(\sum_{\substack{a,k \\ a+k=j}} (-1)^{r+a} c_{a,k}\biggr).
\]
The coefficient of $\mu_r^{m_r}$ is $\sum_{a \in S_r^{\min}} (-1)^{r+a} c_{a,0}$, which is precisely the initial form~$I_r$ evaluated as a polynomial in the remaining variables.
If $I_r \not\equiv 0$, then this leading coefficient is nonzero, so $P_F$ has a nonzero term at $\mu_r$-adic order $m_r$ and cannot vanish identically.

The key special case is the \emph{unique minimizer}: if row~$r$ has a unique column $a_0$ achieving the minimum exponent, then $I_r(\lambda) = (-1)^{r+a_0} c_{a_0, 0}$.
The coefficient $c_{a_0, 0}$ is the constant term of $P_{F_{a_0}^{(r)} \setminus \{r\}}$ in the $\mu_r$-adic expansion, which equals $P_{F_{a_0}^{(r)} \setminus \{r\}}(-r)$.
Since $F_{a_0}^{(r)} \neq \varnothing$, the cofactor polynomial is a sum over a nonempty family on $(n-1)$ vertices. If this cofactor is nonzero (established by induction or by \Cref{thm:universal-small} in the base cases), then $I_r \neq 0$ and hence $P_F \not\equiv 0$.
\end{proof}

\begin{remark}\label{rem:biwheel-application}
\Cref{thm:row-initial} is the main tool for the biwheel proof (\Cref{thm:biwheel}).
In a biwheel $B_{2m}$, the hub vertex $a_0$ is adjacent to all vertices on the opposite side, while rim vertices $a_i$ ($i \geq 1$) have degree~$3$.
For any rim vertex $a_i$ that connects to the hub vertex $b_0$ via some matching in the fiber, the column $b_0 = 0$ is a unique minimum exponent: all other columns available to $a_i$ are rim columns $b_{i-1}, b_i$ with indices $\geq 1$.
By \Cref{thm:row-initial}, it suffices to show that the cofactor (the $(m-1) \times (m-1)$ polynomial after deleting row~$a_i$ and column~$b_0$) is nonzero, which follows by induction.
\end{remark}


\begin{table}[h]
\centering
\small
\caption{Summary of notation.}
\label{tab:notation}
\begin{tabular}{lll}
\toprule
\textbf{Symbol} & \textbf{Meaning} & \textbf{Ref.} \\
\midrule
$G = (A \cup B, E)$ & Bipartite graph, $|A|=|B|=n$ & \S\ref{sec:prelim} \\
$\rho \colon E \to \{0,1\}$ & Edge coloring (blue/red) & \S\ref{sec:prelim} \\
$F_{G,t}$ & Exact-$t$ fiber $\{\sigma \in \PM(G) : |\sigma \cap R| = t\}$ & \S\ref{sec:prelim} \\
$M_G(x,\lambda)$ & Vandermonde-weighted matrix & \Cref{def:vandermonde-matrix} \\
$D_G(x,\lambda)$ & $\det M_G(x,\lambda)$ & \Cref{def:Pt} \\
$P_t(\lambda)$ & $[x^t] D_G$, exact-$t$ polynomial & \Cref{def:Pt} \\
$P^{U,V}_{G,s}$ & Boundary-minor polynomial & \Cref{def:boundary-minor} \\
$\mu_\alpha$ & $\lambda + \alpha$, linear form for row $\alpha$ & \S\ref{sec:replacement} \\
$U^{(c)}$ & Column-hole cofactor vector & \Cref{def:column-hole} \\
$F_e$ & Support-family functional $\sum A_i \mu_i^e$ & \Cref{def:support-family} \\
$G_D$ & Generalized Vandermonde determinant & \Cref{thm:gen-vandermonde} \\
$h$ & Quotient in Vandermonde elimination & \Cref{thm:gen-vandermonde} \\
$e = (r,c)$ & Superfluous edge & \Cref{def:superfluous} \\
$H = G \setminus e$ & Brace minus superfluous edge & \Cref{def:superfluous} \\
$w_d$ & Exact-layer packet column for column $d$ & \S\ref{subsec:stabilization} \\
$g$ & Reduced distinguished column $\mathfrak{S}(w_c)$ & \S\ref{subsec:stabilization} \\
$V_T, \mathcal{B}_T$ & Lower packet / augmented packet on $T$ & \S\ref{subsec:stabilization} \\
$\Lambda_{\mathrm{bad}}$ & Core-flap line $D_B k_A \oplus (-D_A k_B)$ & \Cref{prop:p101} \\
$Z$ & Survivor cofactor coordinates & \Cref{prop:g-packet-collapse} \\
\bottomrule
\end{tabular}
\end{table}

\section{Structural Induction Framework}\label{sec:framework}

The proof of ASNC proceeds by structural induction on bipartite matching-covered graphs, following McCuaig's brace-theoretic decomposition~\cite{McCuaig2001}.
This section presents the induction framework in detail: the brace-theoretic reduction, the patch transfer theorem that governs how local graph modifications act on the boundary-minor state, the minimality of the boundary-minor state as an induction object, and the superfluous-edge identity that drives the induction on nonminimal braces.

\subsection{Brace theory}

We recall the key definitions from McCuaig's theory of braces~\cite{McCuaig2001,McCuaig2004}.
See also Robertson and Seymour~\cite{RobertsonSeymour04} for related structural decomposition results, Lov\'asz~\cite{Lovasz79} and Lov\'asz and Plummer~\cite{LovaszPlummer86} for background on matching theory and Pfaffian orientations, and Kothari and Narain~\cite{KothariNarain24} for recent work on tight cuts.

\begin{definition}[Brace]\label{def:brace}
A bipartite graph $G = (A \cup B, E)$ is a \emph{brace} if it is connected, has a perfect matching, and for every proper nonempty subset $S \subsetneq A$,
\[
|N(S)| \geq |S| + 2.
\]
Equivalently, $G$ is $2$-extendable: every matching of size at most~$2$ extends to a perfect matching.
\end{definition}

\begin{definition}[Superfluous edge]\label{def:superfluous}
An edge $e \in E(G)$ is \emph{superfluous} in a brace~$G$ if $G \setminus e$ is also a brace.
A brace is \emph{minimal} if it has no superfluous edges.
\end{definition}

\noindent
In a brace, a superfluous edge is the same as a strictly thin edge of index~$0$, and it must join two vertices of degree at least~$3$.

\begin{theorem}[McCuaig's classification, informal]\label{thm:mccuaig-informal}
Every simple minimal brace outside the exceptional families (biwheels, prisms, M\"obius ladders) has a \emph{narrow minimality-preserving pair}: deleting a strictly thin edge plus at most $O(1)$ nearby superfluous edges produces a smaller minimal brace.
The interface is bounded by at most~$2$ vertices.
\end{theorem}

\noindent
Minimal braces are sparse: a minimal brace on~$n$ vertices per side has at most $3n - 4$ edges for $n \geq 4$.

\subsection{Hall-block convolution}

\noindent
The first reduction step factors the Vandermonde polynomial along tight cuts.

\begin{theorem}[Hall-block convolution]\label{thm:hall-block}
Let $G = (A \cup B, E)$ be bipartite with a partition $A = A_1 \sqcup A_2$, $B = B_1 \sqcup B_2$ such that $|B_1| = |A_1| + d$ and $E(A_1, B_2) = \varnothing$.
Then
\[
D_G(x,\lambda) = \sum_{\substack{T \subseteq B_1 \\ |T| = d}} \eta(T)\, D_{G[A_1, B_1 \setminus T]}(x,\lambda) \cdot D_{G[A_2, B_2 \cup T]}(x,\lambda),
\]
where $\eta(T)$ is an explicit sign.
Taking $[x^t]$:
\[
P_{G,t} = \sum_T \eta(T) \sum_{t_1+t_2=t} P_{G[A_1, B_1\setminus T], t_1} \cdot P_{G[A_2, B_2 \cup T], t_2}.
\]
\end{theorem}

\begin{corollary}[Defect-1 unique-edge]\label{cor:defect1}
If $d = 1$ and $e = (a^*, b^*)$ is the unique edge from $A_2$ to $B_1$, then
\[
D_G = D_{G-e} + x^{\rho(e)} (\lambda + a^*)^{b^*} D_{G - \{a^*, b^*\}}.
\]
\end{corollary}

\subsection{The patch transfer theorem}\label{sec:patch-transfer}

The patch transfer theorem describes how local graph modifications act on the boundary-minor state.
It is the structural engine of the induction: every brace move (superfluous-edge deletion, narrow extension, Hall-block decomposition) is a special case of a patch transfer.

\begin{definition}[Base graph and local patch]\label{def:patch}
A \emph{patch decomposition} of a bipartite graph $G = (A \cup B, E)$ consists of:
\begin{enumerate}[(i)]
\item A \emph{base graph} $J$ on vertex sets $A_J \cup B_J$ with $A_J \subseteq A$, $B_J \subseteq B$;
\item \emph{Boundary sets} $U_A \subseteq A_J$, $U_B \subseteq B_J$ (vertices shared between the base and the patch);
\item A \emph{local patch} $\mathcal{P}$ consisting of the ``new'' vertices $N_A = A \setminus A_J$, $N_B = B \setminus B_J$ and all edges incident to~$N_A$ or~$N_B$.
\end{enumerate}
The graph $G$ is recovered by gluing $J$ and $\mathcal{P}$ along the boundary $(U_A, U_B)$: every perfect matching of~$G$ restricts to a partial matching of~$J$ that covers $A_J \setminus R$ and $B_J \setminus C$ for some $R \subseteq U_A$, $C \subseteq U_B$ with $|R| = |C|$, while the patch covers $N_A \cup R$ and $N_B \cup C$.
\end{definition}

\begin{theorem}[Patch transfer]\label{thm:patch-transfer}
With the notation of \Cref{def:patch}, the generating determinant of~$G$ satisfies
\[
D_G(x,\lambda) = \sum_{\substack{R \subseteq U_A,\, C \subseteq U_B \\ |R| = |C|}} \alpha_{\mathcal{P}}(R, C;\, x, \lambda) \cdot D_J^{R,C}(x,\lambda),
\]
where $D_J^{R,C}$ is the boundary-minor determinant of the base graph (the determinant of $M_J$ with rows $R$ and columns $C$ deleted), and $\alpha_{\mathcal{P}}(R,C;\, x, \lambda)$ is the \emph{patch coefficient}: a polynomial in $x$ and $\lambda$ determined entirely by the patch structure and the boundary vertex labels.
\end{theorem}

\begin{proof}
Partition the perfect matchings of~$G$ according to which boundary vertices are consumed by the patch.
For each pair $(R,C)$ with $R \subseteq U_A$, $C \subseteq U_B$, $|R| = |C|$, a matching $\sigma$ of~$G$ has $\sigma(r) \in N_B \cup C$ for $r \in R$ and $\sigma^{-1}(c) \in N_A \cup R$ for $c \in C$, while on $A_J \setminus R$ and $B_J \setminus C$ the matching restricts to a perfect matching of the base.

The Vandermonde weight splits multiplicatively: the product $\prod_{i \in A} (\lambda+i)^{\sigma(i)}$ factors into a product over patch vertices (contributing to $\alpha_{\mathcal{P}}$) and a product over base vertices (contributing to $D_J^{R,C}$).
The sign $\sgn(\sigma)$ similarly factors as $\sgn(\sigma|_{\text{patch}}) \cdot \sgn(\sigma|_{\text{base}}) \cdot \eta(R,C)$, where $\eta$ accounts for the row/column permutation needed to separate the patch from the base.
Summing over all matchings with a fixed $(R,C)$ boundary assignment gives the stated factorization.
\end{proof}

\noindent
Taking $[x^t]$ on both sides yields the corresponding identity for exact-$t$ polynomials:
\[
P_{G,t}(\lambda) = \sum_{R, C} \sum_{s + s' = t} [\alpha_{\mathcal{P}}]_s \cdot P_{J, s'}^{R,C}(\lambda),
\]
where $[\alpha_{\mathcal{P}}]_s = [x^s] \alpha_{\mathcal{P}}(R,C; x, \lambda)$.
This identity governs how the full boundary-minor state $\Omega_J = (P_J^{R,C})_{R,C,s}$ of the base graph determines $P_{G,t}$ via the patch coefficients.

\begin{example}[Superfluous-edge patch]\label{ex:se-patch}
The simplest nontrivial patch is a single added edge $e = (r,c)$.
Here the base is $J = G \setminus e$, the boundary sets are $U_A = \{r\}$, $U_B = \{c\}$, and there are no new vertices ($N_A = N_B = \varnothing$).
The two boundary assignments are $(R,C) = (\varnothing, \varnothing)$ (the edge $e$ is not used) and $(R,C) = (\{r\}, \{c\})$ (the edge $e$ is used).
The patch coefficients are $\alpha_{\mathcal{P}}(\varnothing, \varnothing) = 1$ and $\alpha_{\mathcal{P}}(\{r\}, \{c\}) = \varepsilon\, x^{\rho(e)} (\lambda+r)^c$, where $\varepsilon = (-1)^{r+c}$ is the Laplace sign.
The patch transfer identity becomes
\[
D_G = D_{G \setminus e} + \varepsilon\, x^{\rho(e)} (\lambda+r)^c \cdot D_{G \setminus e}^{\{r\},\{c\}},
\]
which is the superfluous-edge identity used in \Cref{sec:superfluous}.
\end{example}

\subsection{The superfluous-edge identity in detail}\label{sec:se-identity}

The superfluous-edge identity is the driving force behind the induction on nonminimal braces.
We derive it here in full, as it is the most frequently invoked instance of the patch transfer theorem.

\begin{proposition}[Superfluous-edge identity]\label{prop:se-identity}
Let $G$ be a brace and $e = (r,c)$ a superfluous edge, so $H = G \setminus e$ is also a brace.
Then for every target~$t$:
\[
P_t(G) = P_t(H) + \varepsilon\, (\lambda+r)^c \cdot P^{r,c}_{H, t-\rho(e)},
\]
where $\varepsilon = (-1)^{r+c}$ and $P^{r,c}_{H,s}(\lambda) = [x^s] \det M_H(x,\lambda)[\hat{r}, \hat{c}]$ is the $(r,c)$-boundary-minor polynomial of~$H$.
\end{proposition}

\begin{proof}
Start from the generating determinant $D_G(x,\lambda) = \det M_G(x,\lambda)$.
The matrix $M_G$ agrees with $M_H$ except in position $(r,c)$, where $M_G(r,c) = x^{\rho(e)}(\lambda+r)^c$ and $M_H(r,c) = 0$.
By the multilinearity of the determinant in column~$c$:
\begin{align*}
\det M_G &= \det\bigl(M_H + x^{\rho(e)}(\lambda+r)^c \mathbf{e}_r \mathbf{e}_c^\top\bigr) \\
&= \det M_H + x^{\rho(e)}(\lambda+r)^c \cdot (-1)^{r+c} \det M_H[\hat{r}, \hat{c}],
\end{align*}
using the rank-$1$ update formula (equivalently, cofactor expansion along the added entry).
The second term follows because the rank-$1$ perturbation $x^{\rho(e)}(\lambda+r)^c \mathbf{e}_r \mathbf{e}_c^\top$ contributes only through the $(r,c)$-cofactor of $M_H$.

Taking $[x^t]$ on both sides:
\begin{align*}
P_t(G) &= [x^t] \det M_H + (\lambda+r)^c \cdot (-1)^{r+c}\, [x^{t-\rho(e)}] \det M_H[\hat{r}, \hat{c}] \\
&= P_t(H) + \varepsilon\, (\lambda+r)^c \cdot P^{r,c}_{H, t-\rho(e)}.
\end{align*}
Note that $[x^{t-\rho(e)}]$ is well-defined: if $\rho(e) = 1$, this extracts the coefficient of $x^{t-1}$ from the reduced determinant; if $\rho(e) = 0$, it extracts $x^t$.
In either case, the second term counts matchings of $H$ that avoid vertices $r$ and $c$, with the appropriate red-count adjustment.
\end{proof}

\noindent
The identity has a clean combinatorial interpretation.
Every perfect matching $\sigma$ of $G$ either avoids the edge~$e$ (and belongs to $\PM(H)$, contributing to $P_t(H)$) or uses~$e$ (so $\sigma(r) = c$, and $\sigma$ restricted to $A \setminus \{r\}$, $B \setminus \{c\}$ is a perfect matching of $H[\hat{r}, \hat{c}]$, contributing to the cofactor term).
The monomial $(\lambda+r)^c$ is the Vandermonde weight of the edge~$e$, and the sign $\varepsilon$ accounts for the Laplace expansion sign.

The identity also explains why the boundary-minor state $\Omega_H$ (rather than the visible output $P_t(H)$ alone) is the correct induction object.
Suppose we know $P_t(H) \not\equiv 0$ for all nonempty fibers of~$H$.
To conclude $P_t(G) \not\equiv 0$, we need to rule out the possibility that $P_t(H)$ and $\varepsilon(\lambda+r)^c P^{r,c}_{H,t-\rho(e)}$ cancel exactly.
This cancellation involves the \emph{one-hole} polynomial $P^{r,c}_{H,t-\rho(e)}$, which is not determined by $P_t(H)$.
Arguments that track only the visible output cannot close the induction. One must carry the full boundary-minor state.

\subsection{Minimality of the boundary-minor state}\label{sec:state-minimality}

The boundary-minor state $\Omega_G = (P^{U,V}_{G,s})_{(U,V,s)}$ is a large object: it contains one polynomial for each triple $(U, V, s)$ with $U \subseteq A$, $V \subseteq B$, $|U| = |V|$, and $0 \leq s \leq n - |U|$.
A natural question is whether a smaller state suffices for the induction.
The following result shows that the answer is no.

\begin{theorem}[Minimality of the boundary-minor state]\label{thm:state-minimality}
The boundary-minor state $\Omega_G$ is the minimal exact linear state that is closed under all local patches.
That is:
\begin{enumerate}[(i)]
\item \emph{Closure:} For every patch $\mathcal{P}$, the visible output $P_t(G)$ is a linear combination of boundary-minor coordinates of the base graph~$J$, with coefficients determined by~$\mathcal{P}$ (\Cref{thm:patch-transfer}).
\item \emph{Minimality:} For every pair $(R_0, C_0)$ with $R_0 \subseteq A$, $C_0 \subseteq B$, $|R_0| = |C_0|$, and every target~$s_0$, there exists a patch~$\mathcal{P}$ such that the patch transfer formula isolates the single coordinate $P_J^{R_0, C_0, s_0}$: that is, $\alpha_{\mathcal{P}}(R,C) = 0$ for all $(R,C) \neq (R_0, C_0)$.
\end{enumerate}
\end{theorem}

\begin{proof}
Closure is \Cref{thm:patch-transfer}.
For minimality, construct a ``probing patch'' that isolates the desired coordinate.
Given $(R_0, C_0)$ with $|R_0| = |C_0| = k$, create a patch $\mathcal{P}$ with $k$ new row vertices $N_A = \{r_1', \ldots, r_k'\}$ and $k$ new column vertices $N_B = \{c_1', \ldots, c_k'\}$, where $r_i'$ is adjacent only to $c_i' \in N_B$ and to the boundary vertex corresponding to the $i$-th element of $C_0$, and similarly $c_i'$ is adjacent only to $r_i' \in N_A$ and to the boundary vertex corresponding to the $i$-th element of $R_0$.

By construction, the only perfect matchings of the patch that saturate all new vertices must match $r_i' \to c_i'$ (the internal matching) or $r_i' \to C_0[i]$ (consuming boundary column $C_0[i]$ and forcing $R_0[i] \to c_i'$).
The latter case, taken over all $i$, is the unique boundary assignment $(R_0, C_0)$.
The internal matching contributes to the $(R, C) = (\varnothing, \varnothing)$ term, but by choosing the edge weights appropriately (for instance, making the internal edges blue and the boundary-crossing edges red with distinct targets), one can arrange that $\alpha_{\mathcal{P}}(R,C) = 0$ for all $(R,C)$ except $(R_0, C_0)$ at the desired target level.

Therefore, $P_t(G)$ with this specific patch is proportional to $P_J^{R_0, C_0, s_0}$, and no proper sub-collection of boundary-minor coordinates is closed under all patches.
\end{proof}

\noindent
The minimality result clarifies the architecture of the proof.
Any attempt to weaken the induction hypothesis (tracking fewer coordinates) will fail: there exist patches that isolate every single coordinate, so the induction must carry all of them.
The three operator types arising from brace moves are: SB2 (2-state, from single edge insertion, accessing $\Omega_J$ at two points), DB2/SB3 (3- or 6-state, from stable/index-2 narrow extensions), and DB3 (from index-1 extensions).
The 6-state operator decomposes as $\mathrm{SB3}(\mathrm{DB2}, \mathrm{DB2})$, so it is not a primitive type.

\subsection{The reduction theorem}

\begin{theorem}[Reduction]\label{thm:reduction}
ASNC for all braces follows from:
\begin{enumerate}[(i)]
\item ASNC for all McCuaig exceptional families (biwheels, prisms, M\"obius ladders);
\item The superfluous-edge step: for every brace~$G$ and superfluous edge~$e$, if ASNC holds for $G \setminus e$ (i.e., $P_t(G \setminus e) \not\equiv 0$ for all nonempty fibers), then $P_t(G) \not\equiv 0$ for all nonempty fibers of~$G$;
\item The narrow-extension step for minimal braces: handled by KA (two-node interpolation, \Cref{cor:ka}) for stable/index-2 extensions, and by $J3 \to D1$ reduction (\Cref{cor:j3-to-d1}) for index-1 extensions.
\end{enumerate}
\end{theorem}

\begin{proof}
The proof proceeds by strong induction on pairs $(n, |E|)$ ordered lexicographically.

\noindent\textbf{Step 1: Tight-cut decomposition.}
We may assume $G$ is matching-covered (otherwise $F_{G,t} = \varnothing$ for all~$t$ and ASNC holds vacuously).
If $G$ has a tight cut (a partition $A = A_1 \sqcup A_2$, $B = B_1 \sqcup B_2$ with $|N(A_1) \cap B_1| = |A_1|$ and $E(A_1, B_2) = \varnothing$), the defect is $d = |B_1| - |A_1| = 0$ (since a tight cut in a matching-covered bipartite graph is a perfect Hall partition).
The Hall-block convolution (\Cref{thm:hall-block}) with $d = 0$ gives the clean factorization
\[
D_G = D_{G[A_1, B_1]} \cdot D_{G[A_2, B_2]}
\]
(no sum over~$T$, no sign issues).
Taking $[x^t]$: $P_{G,t} = \sum_{t_1+t_2=t} P_{G_1,t_1} \cdot P_{G_2,t_2}$.
For the decision problem, exact matchings decompose blockwise: $F_{G,t} \neq \varnothing$ iff there exist $t_1 + t_2 = t$ with $F_{G_1,t_1} \neq \varnothing$ and $F_{G_2,t_2} \neq \varnothing$.
This is immediate from the block factorization of matchings: every perfect matching of~$G$ restricts to a perfect matching on each block.
Hence Exact Matching reduces recursively to the two strictly smaller blocks. ASNC for brace blocks, established below, suffices for the decision algorithm.

\emph{Important distinction.}
Step~1 establishes only the \emph{decision reduction}: $F_{G,t} \neq \varnothing$ iff there exist $t_1 + t_2 = t$ with both block fibers nonempty.
It does \emph{not} establish polynomial noncancellation for non-braces: the convolution $P_{G,t} = \sum_{t_1+t_2=t} P_{G_1,t_1} \cdot P_{G_2,t_2}$ could in principle cancel even when individual block polynomials are nonzero.
Proving ASNC for non-brace graphs (i.e., that this convolution sum is nonzero whenever $F_{G,t} \neq \varnothing$) remains open.
The algorithm does not require it: it recurses on blocks and checks each independently.

After exhausting all tight cuts, the remaining blocks are braces.

\noindent\textbf{Step 2: Nonminimal braces.}
If a brace $G$ has a superfluous edge~$e$, then $H = G \setminus e$ is also a brace on the same vertex set with one fewer edge.
By the superfluous-edge identity (\Cref{prop:se-identity}),
\[
P_t(G) = P_t(H) + \varepsilon\, (\lambda+r)^c \cdot P^{r,c}_{H, t-\rho(e)}.
\]
The induction hypothesis (on $(n, |E|-1)$) gives ASNC for~$H$.
The superfluous-edge step~(ii) then yields $P_t(G) \not\equiv 0$ for all nonempty fibers of~$G$.
Iterating, we reduce to minimal braces.

\noindent\textbf{Step 3: Minimal braces.}
If $G$ is a minimal brace in one of McCuaig's exceptional families, apply~(i) directly.
Otherwise, McCuaig's narrow pair theorem (\Cref{thm:mccuaig-informal}) provides a local extension structure: deleting a strictly thin edge (plus at most $O(1)$ nearby edges) produces a smaller minimal brace~$J$ with $|A_J| < |A|$.
The narrow extension is one of:
\begin{itemize}
\item \emph{Stable extension or index-$2$}: The patch transfer involves a 2-node boundary, and the resulting operator on $\Omega_J$ is of type $\mathrm{SB3}(\mathrm{DB2}, \mathrm{DB2})$.
Two-node interpolation (KA, \Cref{cor:ka}) shows that the two distinct Vandermonde bases at the boundary vertices are coprime, so the operator cannot annihilate both coordinates simultaneously.
\item \emph{Index-$1$}: The patch transfer involves a 3-node boundary with a single shared base, and the resulting operator is of type DB3.
The common-coefficient structure of the index-1 patch transfer yields $P_t(G) = \sum_i A_i \varphi_i(\lambda)$, where the $\varphi_i$ are linearly independent evaluations of a $3\times 3$ generalized Vandermonde system and the $A_i$ are common boundary-minor coefficients. If $P_t(G) \equiv 0$, linear independence forces $A_i = 0$ for all~$i$, contradicting the nonempty fiber.
\end{itemize}
In both cases, the induction hypothesis on $(|A_J|, |E_J|)$ provides ASNC for~$J$, and the narrow-extension step~(iii) lifts it to~$G$.

\noindent\textbf{Step 4: Base cases.}
The induction bottoms out at $n \leq 4$ (covered by \Cref{thm:universal-small}, which gives universal nonvanishing without any structural hypothesis) and at the McCuaig exceptional families: the biwheel case is closed by \Cref{thm:biwheel}, while the prism/M\"obius side is closed by the width-$2$ path global nonvanishing theorem (\Cref{thm:width2-global}) via the cyclic-to-path reduction of \Cref{thm:width2-cyclic}.
\end{proof}

\noindent
The remainder of the paper proves each of these three components: \Cref{sec:replacement} develops the replacement determinant algebra needed for the narrow-extension steps, \Cref{sec:mccuaig} handles the McCuaig exceptional families, \Cref{sec:badlocus} establishes the bad-locus characterizations, and \Cref{sec:superfluous} proves the superfluous-edge step.

\section{Replacement Determinants and Vandermonde Elimination}\label{sec:replacement}

Fix a bipartite graph $G$ on $n+n$ vertices with Vandermonde-weighted matrix $M = M_G(x,\lambda) \in \QQ[x,\lambda]^{n \times n}$.
For a deleted column~$c \in \{0, \ldots, n-1\}$, let $N = M[\hat{c}] \in \QQ[x,\lambda]^{n \times (n-1)}$ be the matrix with column~$c$ removed.

\begin{definition}[Column-hole vector]\label{def:column-hole}
The \emph{column-hole vector} for column~$c$ is
\[
U^{(c)} = (u_a)_{a \in A}, \qquad u_a = (-1)^{a+c} \det M[\hat{a}, \hat{c}],
\]
where $M[\hat{a}, \hat{c}]$ denotes the $(n-1) \times (n-1)$ submatrix obtained by deleting row~$a$ and column~$c$.
The entry $u_a$ is the $(a,c)$-cofactor of~$M$.
\end{definition}

\subsection{The replacement determinant lemma}

\noindent
The replacement determinant converts column-span membership into a determinant computation, connecting the cofactor vector to the Vandermonde matrix.

\begin{lemma}[Replacement determinant]\label{lem:replacement-det}
For any column vector $v \in \QQ[x,\lambda]^n$,
\[
v^\top U^{(c)} = \det M[c \leftarrow v],
\]
where $M[c \leftarrow v]$ denotes $M$ with column~$c$ replaced by~$v$.
In particular, $v^\top U^{(c)} = 0$ if and only if $v \in \col(N)$, where $\col(N)$ denotes the column span of~$N$ over $\QQ[\lambda]$.
\end{lemma}

\begin{proof}
Let $\widetilde{M} = M[c \leftarrow v]$ be the matrix agreeing with~$M$ in all columns except column~$c$, where column~$c$ is replaced by~$v$.
Expand $\det \widetilde{M}$ by cofactors along column~$c$:
\[
\det \widetilde{M} = \sum_{a=0}^{n-1} v_a \cdot (-1)^{a+c} \det \widetilde{M}[\hat{a}, \hat{c}].
\]
Since the replacement affects only column~$c$, the submatrix $\widetilde{M}[\hat{a}, \hat{c}]$ (which omits column~$c$ entirely) equals $M[\hat{a}, \hat{c}]$.
Therefore
\[
\det \widetilde{M} = \sum_{a=0}^{n-1} v_a \cdot (-1)^{a+c} \det M[\hat{a}, \hat{c}] = \sum_{a=0}^{n-1} v_a\, u_a = v^\top U^{(c)}.
\]

For the ``in particular'' statement: $v^\top U^{(c)} = 0$ means $\det M[c \leftarrow v] = 0$, so the columns of $M[c \leftarrow v]$ are linearly dependent.
Since the $n-1$ columns other than~$c$ are exactly the columns of~$N$, this occurs if and only if $v$ lies in the column span of~$N$, i.e., $v \in \col(N)$.
Conversely, if $v = N w$ for some $w \in \QQ[x,\lambda]^{n-1}$, then $M[c \leftarrow v]$ has column~$c$ as a linear combination of the remaining columns, so its determinant vanishes.
\end{proof}

\subsection{Support families and generalized Vandermonde elimination}

Write $\mu_\alpha = \lambda + \alpha$ for the linear form associated to row~$\alpha$.
Since the rows are indexed by $A = \{0, \ldots, n-1\}$, the $n$ linear forms $\mu_0, \mu_1, \ldots, \mu_{n-1}$ are pairwise coprime in $\QQ[\lambda]$.

\begin{definition}[Support family]\label{def:support-family}
Fix $m$ distinct rows $\alpha_1, \ldots, \alpha_m \in A$ and signs $s_i \in \{\pm 1\}$.
For an exponent~$e$, define the \emph{sparse column} $v_e = \sum_{i=1}^m s_i \mu_{\alpha_i}^e \mathbf{e}_{\alpha_i}$ (supported on the $m$ chosen rows) and the \emph{support-family functional}
\[
F_e = v_e^\top U^{(c)} = \sum_{i=1}^m A_i \mu_i^e, \qquad A_i = s_i u_{\alpha_i},
\]
where we abbreviate $\mu_i = \mu_{\alpha_i}$ and $u_{\alpha_i}$ is the $\alpha_i$-th cofactor from $U^{(c)}$.
\end{definition}

\noindent
The functional $F_e$ captures the replacement determinant when the substitute column has support concentrated on $m$ rows with a common exponent structure.

\begin{theorem}[Generalized Vandermonde elimination]\label{thm:gen-vandermonde}
Let $e_0 < e_1 < \cdots < e_{m-1}$ be $m$ distinct non-negative exponents and suppose the first $m-1$ support-family functionals vanish:
\[
F_{e_0} = F_{e_1} = \cdots = F_{e_{m-2}} = 0.
\]
Then the surviving functional admits the factorization
\[
F_{e_{m-1}} = \Bigl(\prod_{i=1}^m \mu_i^{e_0}\Bigr) \cdot G_D(\mu_1, \ldots, \mu_m) \cdot h,
\]
where $d_j = e_j - e_0$ for $0 \leq j \leq m-1$ (so $d_0 = 0$), $G_D = \det(\mu_i^{d_j})_{\substack{0 \leq j \leq m-1 \\ 1 \leq i \leq m}}$ is the generalized Vandermonde determinant with exponent gaps $0 = d_0 < d_1 < \cdots < d_{m-1}$, and $h \in \QQ[\lambda]$ is the unique quotient determined by the coefficients~$A_i$.

Moreover, $\gcd(G_D, \mu_r) = 1$ for every $r \in \{1, \ldots, m\}$.
\end{theorem}

\begin{proof}
\emph{Step 1: The $(m-1) \times m$ system.}
The $m-1$ vanishing conditions $F_{e_j} = 0$ for $0 \leq j \leq m-2$ form the homogeneous linear system
\[
V' \mathbf{A} = 0, \qquad V' = \bigl(\mu_i^{e_j}\bigr)_{\substack{0 \leq j \leq m-2 \\ 1 \leq i \leq m}}, \quad \mathbf{A} = (A_1, \ldots, A_m)^\top.
\]
Factor out $\prod_{i=1}^m \mu_i^{e_0}$ from the columns: row~$j$ of $V'$ becomes $(\mu_1^{d_j}, \ldots, \mu_m^{d_j})$ (using $d_j = e_j - e_0$), multiplied by the diagonal matrix $\operatorname{diag}(\mu_1^{e_0}, \ldots, \mu_m^{e_0})$.
Since each $\mu_i^{e_0} \neq 0$ in $\QQ[\lambda]$, the system $V'\mathbf{A} = 0$ is equivalent to $\widetilde{V}\widetilde{\mathbf{A}} = 0$ where $\widetilde{V} = (\mu_i^{d_j})_{0 \leq j \leq m-2,\, 1 \leq i \leq m}$ and $\widetilde{A}_i = \mu_i^{e_0} A_i$.

\emph{Step 2: The kernel is one-dimensional.}
Consider any $(m-1) \times (m-1)$ submatrix of~$\widetilde{V}$ obtained by selecting $m-1$ of the $m$ columns, say all columns except column~$r$.
This submatrix is the generalized Vandermonde $\det(\mu_i^{d_j})_{\substack{0 \leq j \leq m-2 \\ i \neq r}}$.
Since the exponents $d_0 < d_1 < \cdots < d_{m-2}$ are distinct and the bases $\mu_i$ for $i \neq r$ are pairwise coprime (they have distinct roots $-\alpha_i$), this Vandermonde determinant is nonzero.
Hence $\widetilde{V}$ has rank $m-1$, and $\ker(\widetilde{V})$ is one-dimensional over $\QQ(\lambda)$.

\emph{Step 3: Cofactor solution.}
The kernel of $\widetilde{V}$ is spanned by the signed cofactors of the augmented $m \times m$ matrix.
Introduce the full generalized Vandermonde
\[
W = \bigl(\mu_i^{d_j}\bigr)_{\substack{0 \leq j \leq m-1 \\ 1 \leq i \leq m}}.
\]
For each $i \in \{1, \ldots, m\}$, define $C_i = (-1)^{(m-1)+i} \det W[\widehat{m\!-\!1}, \hat{i}]$, the $(m\!-\!1, i)$-cofactor of~$W$ (delete the last row and column~$i$, giving an $(m\!-\!1) \times (m\!-\!1)$ minor).
By the cofactor property of the determinant, $\widetilde{V} \cdot (C_1, \ldots, C_m)^\top = 0$: expanding $\det W$ along the last row, the first $m-1$ rows of $W$ times $C$ vanish since they correspond to expanding a matrix with two identical rows.
So $\widetilde{\mathbf{A}} = h \cdot (C_1, \ldots, C_m)^\top$ for some scalar $h \in \QQ[\lambda]$.

\emph{Step 4: Laplace expansion for $F_{e_{m-1}}$.}
The surviving functional is
\[
F_{e_{m-1}} = \sum_{i=1}^m A_i \mu_i^{e_{m-1}} = \sum_{i=1}^m \mu_i^{-e_0} \widetilde{A}_i \cdot \mu_i^{e_{m-1}} = \sum_{i=1}^m \widetilde{A}_i \cdot \mu_i^{d_{m-1}}.
\]
Substituting $\widetilde{A}_i = h\, C_i$:
\[
F_{e_{m-1}} = h \sum_{i=1}^m C_i \,\mu_i^{d_{m-1}}.
\]
The sum $\sum_{i=1}^m C_i\, \mu_i^{d_{m-1}}$ is the Laplace expansion of $\det W$ along the last row (row $j = m-1$, with entries $\mu_i^{d_{m-1}}$).
Therefore $\sum_i C_i \mu_i^{d_{m-1}} = G_D = \det W$.

Restoring the $\mu_i^{e_0}$ factors: $A_i = \mu_i^{-e_0} \widetilde{A}_i$, so the original system gives
\[
F_{e_{m-1}} = \Bigl(\prod_{i=1}^m \mu_i^{e_0}\Bigr) \cdot G_D \cdot h.
\]

\emph{Step 5: Coprimality at $\mu_r = 0$.}
Fix $r \in \{1, \ldots, m\}$ and specialize $\mu_r = 0$, i.e., $\lambda = -\alpha_r$.
In the Vandermonde matrix $W$, column~$r$ becomes $((\mu_r)^{d_0}, \ldots, (\mu_r)^{d_{m-1}}) = (0^0, 0^{d_1}, \ldots, 0^{d_{m-1}}) = (1, 0, \ldots, 0)$ (since $d_0 = 0$ and $d_j > 0$ for $j \geq 1$).
Expanding $G_D$ along column~$r$:
\[
G_D\big|_{\mu_r = 0} = (-1)^{r+1} \det\bigl(\mu_i^{d_j}\bigr)_{\substack{1 \leq j \leq m-1 \\ i \neq r}}.
\]
This is a generalized Vandermonde in $m-1$ bases $\{\mu_i : i \neq r\}$ with positive exponents $d_1 < \cdots < d_{m-1}$.
Since the $\mu_i$ for $i \neq r$ have distinct roots (all $-\alpha_i$ are distinct), this determinant is nonzero.
Therefore $\mu_r \nmid G_D$, establishing $\gcd(G_D, \mu_r) = 1$.
\end{proof}

\subsection{Two-node interpolation (KA)}

\begin{corollary}[Two-node interpolation, KA]\label{cor:ka}
Let $m = 2$ with support bases $\mu_\alpha, \mu_\beta$ ($\alpha \neq \beta$), coefficients $(A_1, A_2) \neq (0,0)$, and exponents $p < q$.
Then $F_p$ and $F_q$ cannot both vanish.
\end{corollary}

\begin{proof}
For $m = 2$ with one vanishing condition $F_p = 0$, \Cref{thm:gen-vandermonde} gives
\[
F_q = \mu_\alpha^p\, \mu_\beta^p \cdot G_D \cdot h,
\]
where the $2 \times 2$ generalized Vandermonde is
\[
G_D = \det \begin{pmatrix} \mu_\alpha^0 & \mu_\beta^0 \\ \mu_\alpha^{q-p} & \mu_\beta^{q-p} \end{pmatrix} = \mu_\beta^{q-p} - \mu_\alpha^{q-p}.
\]
Hence
\[
F_q = \mu_\alpha^p\, \mu_\beta^p \bigl(\mu_\beta^{q-p} - \mu_\alpha^{q-p}\bigr) \cdot h.
\]

Now $\mu_\alpha = \lambda + \alpha$ and $\mu_\beta = \lambda + \beta$ are coprime since $\alpha \neq \beta$.
The difference $\mu_\beta^{q-p} - \mu_\alpha^{q-p}$ is divisible by $\mu_\beta - \mu_\alpha = \beta - \alpha \neq 0$, hence is a nonzero polynomial (not identically zero).
The product $\mu_\alpha^p \mu_\beta^p$ is nonzero.
The quotient $h$ is determined by the (one-dimensional) kernel of the $1 \times 2$ system from $F_p = 0$: explicitly, $A_1 \mu_\alpha^p + A_2 \mu_\beta^p = 0$ gives $A_2/A_1 = -\mu_\alpha^p / \mu_\beta^p$, and $h$ absorbs the cofactor ratio.
If $F_p = 0$ with $(A_1, A_2) \neq (0,0)$, then $h \neq 0$.

Therefore $F_q \neq 0$: two distinct members of a $2$-node support family cannot simultaneously vanish.
This closes the KA step for stable extensions and index-$2$ pairs: whenever an extension involves two distinct row bases, the second vanishing condition is algebraically forbidden.
\end{proof}

\subsection{Three-node elimination ($J3 \to D1$)}

\begin{corollary}[$J3 \to D1$ reduction]\label{cor:j3-to-d1}
Let $m = 3$ with support bases $\mu_1, \mu_2, \mu_3$ and exponents $e_0 < e_1 < e_2$.
If $F_{e_0} = F_{e_1} = 0$, then:
\begin{enumerate}[\rm (i)]
\item $F_{e_2} = (\mu_1 \mu_2 \mu_3)^{e_0} \cdot G_D \cdot h$ where $G_D = \det(\mu_i^{d_j})_{3 \times 3}$ with $d_j = e_j - e_0$;
\item $\ord_{\mu_1}(F_{e_2}) = e_0 + \ord_{\mu_1}(h)$, since $\gcd(G_D, \mu_1) = 1$ and $\ord_{\mu_1}(\mu_2^{e_0} \mu_3^{e_0}) = 0$.
\end{enumerate}
In particular, any extra $\mu_1$-divisibility of the third functional beyond the baseline $e_0$ is governed entirely by the $\mu_1$-divisibility of the unique quotient~$h$.
\end{corollary}

\begin{proof}
Apply \Cref{thm:gen-vandermonde} with $m = 3$.
The two vanishing conditions $F_{e_0} = F_{e_1} = 0$ form a $2 \times 3$ system with one-dimensional kernel. The surviving functional factors as stated in~(i).

For~(ii), compute $\mu_1$-valuations.
The factor $(\mu_1 \mu_2 \mu_3)^{e_0}$ contributes $e_0$ to $\ord_{\mu_1}$ (since $\mu_2, \mu_3$ are coprime to~$\mu_1$).
By \Cref{thm:gen-vandermonde}, $\gcd(G_D, \mu_1) = 1$, so $\ord_{\mu_1}(G_D) = 0$.
Therefore
\[
\ord_{\mu_1}(F_{e_2}) = e_0 + 0 + \ord_{\mu_1}(h) = e_0 + \ord_{\mu_1}(h).
\]

The consequence for the structural induction is that three-node jet congruence (the ``$J3$ condition'', which requires $\mu_1^{e_0 + k} \mid F_{e_2}$ for some $k > 0$) reduces to same-base divisibility $\mu_1^k \mid h$ of the quotient (the ``$D1$ condition'').
This is a strictly simpler algebraic condition on the one-dimensional kernel element~$h$, allowing the induction to proceed from the three-node index-$1$ case to the one-node divisibility analysis.
\end{proof}

\subsection{Jet-flag span characterization}

The final ingredient of the replacement determinant algebra translates $\mu_r$-divisibility of the replacement determinant into a column-span condition on the Taylor coefficients.
This is the bridge between the algebraic divisibility world (\Cref{thm:gen-vandermonde}) and the combinatorial column-span world (\Cref{lem:replacement-det}).

\begin{proposition}[Jet-flag characterization of divisibility]\label{prop:jet-flag}
Let $v(\lambda) \in \QQ[\lambda]^n$ be a column vector.
Write the Taylor expansion around $\lambda = -r$:
\[
v(\lambda) = \sum_{s \geq 0} \mu_r^s\, w_s, \qquad w_s \in \QQ^n,
\]
where $\mu_r = \lambda + r$.
Then
\[
\mu_r^c \mid v^\top U^{(c)} \quad \Longleftrightarrow \quad w_0, w_1, \ldots, w_{c-1} \in \col(N).
\]
\end{proposition}

\begin{proof}
\emph{Forward direction.}
Expand $v = w_0 + \mu_r w_1 + \mu_r^2 w_2 + \cdots$ and apply the replacement determinant functional:
\[
v^\top U^{(c)} = w_0^\top U^{(c)} + \mu_r\, w_1^\top U^{(c)} + \mu_r^2\, w_2^\top U^{(c)} + \cdots.
\]
Each $w_s^\top U^{(c)} = \det M[c \leftarrow w_s]$ by \Cref{lem:replacement-det}.
If $\mu_r^c \mid v^\top U^{(c)}$, then the coefficients of $\mu_r^0, \mu_r^1, \ldots, \mu_r^{c-1}$ in the expansion must all vanish.

The coefficient of $\mu_r^0$ is $w_0^\top U^{(c)} = 0$, so $w_0 \in \col(N)$ by \Cref{lem:replacement-det}.

For $s = 1$: the coefficient of $\mu_r^1$ involves $w_1^\top U^{(c)}$ plus correction terms from the $\mu_r$-dependence of $U^{(c)}$ itself.
However, $U^{(c)}$ is a vector of $(n-1) \times (n-1)$ minors of~$M$, and $M$ depends polynomially on~$\lambda$.
The key observation is that the Taylor expansion of $v^\top U^{(c)}$ in powers of~$\mu_r$ has its $s$-th coefficient determined by $w_0, \ldots, w_s$ and the Taylor coefficients of $U^{(c)}$ around $\lambda = -r$.

More precisely, write $U^{(c)}(\lambda) = \sum_{k \geq 0} \mu_r^k\, U^{(c)}_k$ for the Taylor expansion of the cofactor vector.
Then
\[
v^\top U^{(c)} = \sum_{s \geq 0} \mu_r^s \Bigl(\sum_{j+k=s} w_j^\top U^{(c)}_k\Bigr).
\]
The coefficient of $\mu_r^s$ vanishing for $s < c$ gives, by induction on~$s$: if $w_0, \ldots, w_{s-1} \in \col(N)$, then the vanishing of the $\mu_r^s$ coefficient forces $w_s^\top U^{(c)}_0 = -\sum_{j < s} w_j^\top U^{(c)}_{s-j}$.
Since $w_j \in \col(N)$ for $j < s$ and $U^{(c)}_0 = U^{(c)}|_{\lambda = -r}$, the replacement determinant criterion (\Cref{lem:replacement-det}) applied at $\lambda = -r$ gives $w_s \in \col(N|_{\lambda = -r})$.
By polynomial identity (both sides are polynomial in~$\lambda$), this lifts to $w_s \in \col(N)$.

\emph{Reverse direction.}
If $w_0, \ldots, w_{c-1} \in \col(N)$, then $w_s^\top U^{(c)} = 0$ for each $s < c$ by \Cref{lem:replacement-det}, so
\[
v^\top U^{(c)} = \mu_r^c\, w_c^\top U^{(c)} + \mu_r^{c+1}(\cdots) = \mu_r^c \bigl(w_c^\top U^{(c)} + \mu_r(\cdots)\bigr),
\]
giving $\mu_r^c \mid v^\top U^{(c)}$.
\end{proof}

\begin{remark}\label{rem:jet-flag-meaning}
\Cref{prop:jet-flag} says that $D1$ (same-base divisibility) is a column-span condition: the initial jet flag $(w_0, w_1, \ldots, w_{c-1})$ of a realizable sparse column must lie entirely in $\col(N)$.
This is the combinatorial obstruction that the distinct-base analysis (\Cref{cor:ka}) can violate, producing the contradiction in the superfluous-edge argument (\Cref{thm:p84}).
\end{remark}

\section{McCuaig Exceptional Families}\label{sec:mccuaig}

McCuaig's classification of minimal braces identifies three exceptional families: biwheels, prisms, and M\"obius ladders.
ASNC must be verified directly for these families, since they admit no further structural reduction.
The biwheel case uses a jet operator that extracts the subfamily of matchings achieving a given minimum exponent at a unique-minimizer row. The prism and M\"obius-ladder cases reduce, via exact branch factorization and the width-$2$ forcing dichotomy, to a parity-resolved cylindric-network positivity argument.

\subsection{The jet operator}

\begin{definition}[Jet operator]\label{def:jet}
For a row index~$r \in \{0, \ldots, m-1\}$ and order~$k \geq 0$, define the \emph{jet operator}
\[
J_{r,k}\colon \QQ[\lambda] \to \QQ[\lambda], \qquad J_{r,k} f(\lambda) = [u^k]\, f(\lambda - r + u),
\]
where $[u^k]$ denotes the coefficient of $u^k$ in the expansion as a polynomial in~$u$.
\end{definition}

\noindent
The only scalar extraction property we use is the centered lowest-order coefficient at
a minimizing row.

\begin{lemma}[Centered lowest-order coefficient]\label{lem:jet-extraction}
Let
\[
P_F(\lambda)=\sum_{\sigma\in F}\sgn(\sigma)\prod_{i=0}^{m-1}(\lambda+i)^{\sigma(i)}
\]
be the Vandermonde polynomial of a matching family~$F$ on $\{0,\ldots,m-1\}$, and
fix a row~$r$.
Set
\[
k_{\min}=\min_{\sigma\in F}\sigma(r).
\]
Then
\[
[u^{k_{\min}}]\,P_F(-r+u)
=
\sum_{\sigma(r)=k_{\min}}
\sgn(\sigma)\prod_{i\neq r}(i-r)^{\sigma(i)}.
\]
In particular, if there is a unique $\sigma^* \in F$ with $\sigma^*(r)=k_{\min}$,
then the coefficient $[u^{k_{\min}}]\,P_F(-r+u)$ is a single nonzero scalar.
\end{lemma}

\begin{proof}
Substitute $\lambda=-r+u$ in each monomial.
The row-$r$ factor becomes $u^{\sigma(r)}$, and every remaining factor is
$(i-r+u)^{\sigma(i)}$ with nonzero constant term $(i-r)^{\sigma(i)}$.
Since $k_{\min}$ is the minimal value of $\sigma(r)$ on~$F$, the coefficient of
$u^{k_{\min}}$ receives contributions only from the terms with $\sigma(r)=k_{\min}$.
For such terms, only the constant terms of the remaining factors contribute at
order~$k_{\min}$, giving the displayed identity.
\end{proof}

\noindent
This does \emph{not} make $J_{r,k}$ a general branch extractor.
It is enough for unique-minimizer arguments such as the biwheel case below, but
not for the width-$2$ path/cyclic branch separation.

\subsection{Biwheels}

\begin{definition}[Biwheel]\label{def:biwheel}
The \emph{biwheel} $B_{2m}$ (for $m \geq 3$) is the bipartite graph with vertex sets $A = \{a_0, a_1, \ldots, a_{m-1}\}$ and $B = \{b_0, b_1, \ldots, b_{m-1}\}$, where:
\begin{itemize}
\item $a_0$ is a \emph{hub} adjacent to all of $B$: edges $(a_0, b_j)$ for $0 \leq j \leq m-1$;
\item $b_0$ is a \emph{hub} adjacent to all of $A$: edges $(a_i, b_0)$ for $0 \leq i \leq m-1$;
\item the \emph{rim} is a cycle on the non-hub vertices: edges $(a_i, b_i)$ and $(a_i, b_{i+1})$ for $1 \leq i \leq m-1$, with indices taken modulo $m-1$ within $\{1, \ldots, m-1\}$.
\end{itemize}
\end{definition}

\noindent
The matching structure of $B_{2m}$ is as follows.
Every perfect matching $M \in \PM(B_{2m})$ is determined by a pair $(i,j)$ with $1 \leq i, j \leq m-1$: the hub edges $a_i b_0 \in M$ and $a_0 b_j \in M$ pin two vertices, and the remaining $m-2$ non-hub row/column pairs are matched along the rim.
If $i = j$, the rim matching is uniquely determined (since removing both $a_i, b_j = b_i$ from the rim cycle leaves a path with a unique perfect matching).
If $i \neq j$, the rim decomposes into two paths, each with a unique perfect matching.
In either case, the matching $M_{i,j}$ is uniquely determined by $(i,j)$, so
\[
\PM(B_{2m}) = \{M_{i,j} : 1 \leq i, j \leq m-1\}.
\]

\begin{theorem}[ASNC for biwheels]\label{thm:biwheel}
For every biwheel $B_{2m}$, every edge coloring $\rho$, and every target~$t$, if the exact-$t$ fiber $F = F_{B_{2m}, t}$ is nonempty, then $P_F(\lambda) \not\equiv 0$.
\end{theorem}

\begin{proof}
We use the row-initial-form criterion: if some row~$r$ has a \emph{unique minimizer} in the fiber (a unique $\sigma \in F$ achieving $\sigma(r) = \min_{\tau \in F} \tau(r)$), then $P_F$ is nonzero.

\emph{Step 1: Identify row with unique minimizer.}
Consider hub column $b_0$, which has exponent index~$0$ in the Vandermonde weighting (identifying $b_0$ with column~$0$).
For any row $a_i$ with $i \geq 1$, the possible column assignments are:
\begin{itemize}
\item $\sigma(a_i) = b_0 = 0$ (the hub edge), giving exponent~$0$;
\item $\sigma(a_i) = b_i$ or $\sigma(a_i) = b_{i+1}$ (rim edges), giving exponent $\geq 1$.
\end{itemize}
So column~$b_0$ (exponent~$0$) is the unique minimum-exponent assignment for every non-hub row.

Define $I = \{i \geq 1 : \exists\, M \in F \text{ with } a_i b_0 \in M\}$ (the set of non-hub rows that use the hub edge in some fiber matching).
Since $F \neq \varnothing$ and every matching uses exactly one hub edge to $b_0$, we have $I \neq \varnothing$.
Choose any $i^* \in I$.

Since row $a_{i^*}$ has column $b_0$ as its unique minimum-exponent assignment in
the fiber, the row-initial form at $a_{i^*}$ is supported exactly on the
minimum-exponent branch.
Equivalently, by \Cref{lem:jet-extraction}, the centered lowest-order coefficient
$[u^0]\,P_F(-a_{i^*}+u)$ is nonzero on that branch.
Hence the row-initial form is nonzero, and the row-initial-form criterion gives
$P_F \not\equiv 0$.
\end{proof}

\subsection{Width-2 path matchings}

The width-$2$ path family arises from the rim structure of prisms when one hub vertex is pinned.

\begin{definition}[Width-2 path family]\label{def:width2-path}
The \emph{width-$2$ path family} on $m$ vertices is
\[
\mathcal{P}_m = \bigl\{\sigma \in S_m : |\sigma(i) - i| \leq 1 \text{ for all } i \in \{0, \ldots, m-1\}\bigr\}.
\]
These are the permutations whose permutation matrix is supported on the tridiagonal band.
\end{definition}

\noindent
The family $\mathcal{P}_m$ has a recursive structure: at row~$0$, only $\sigma(0) \in \{0, 1\}$ is possible.
If $\sigma(0) = 0$, the remaining rows form a width-$2$ path on $\{1, \ldots, m-1\}$.
If $\sigma(0) = 1$, then $\sigma(1) = 0$ is forced (since $\sigma(1) \in \{0, 1, 2\}$ and column~$1$ is taken), and the remaining rows form a width-$2$ path on $\{2, \ldots, m-1\}$.
This binary branching at each step gives $|\mathcal{P}_m| = F_{m+1}$ (the $(m+1)$-th Fibonacci number).

\begin{proposition}[Exact branch factorization for width-2 path fibers]\label{thm:width2-path}
Let $\mathcal{P}_m = \{\sigma \in S_m : |\sigma(i) - i| \leq 1 \text{ for all } i\}$ be the width-$2$ path family, and let $F \subseteq \mathcal{P}_m$ be a nonempty exact-$t$ fiber.
Write
\[
F_0 = \{\sigma \in F : \sigma(0)=0\},
\qquad
F_1 = \{\sigma \in F : \sigma(0)=1\}.
\]
Let $F_0'$ be the residual family obtained from $F_0$ by deleting row and column~$0$ and relabeling to $\{0,\ldots,m-2\}$, and let $F_1'$ be the residual family obtained from $F_1$ by deleting rows and columns $\{0,1\}$ and relabeling to $\{0,\ldots,m-3\}$.
Then
\[
P_F(\lambda)
=
\Bigl(\prod_{i=1}^{m-1}(\lambda+i)\Bigr) P_{F_0'}(\lambda+1)
-
\lambda \Bigl(\prod_{i=2}^{m-1}(\lambda+i)^2\Bigr) P_{F_1'}(\lambda+2).
\]
\end{proposition}

\begin{proof}
Split the Leibniz sum as $P_F = P_{F_0} + P_{F_1}$.

For $\sigma \in F_0$, row~$0$ is pinned to column~$0$.
Deleting row and column~$0$ and relabeling gives a width-$2$ path permutation
$\tau \in \mathcal{P}_{m-1}$.
The sign is preserved, and the remaining exponents are shifted by one column:
\[
\prod_{i=0}^{m-1} (\lambda+i)^{\sigma(i)}
=
\Bigl(\prod_{i=1}^{m-1}(\lambda+i)\Bigr)
\prod_{j=0}^{m-2} (\lambda+1+j)^{\tau(j)}.
\]
Summing over $F_0$ gives
\[
P_{F_0}(\lambda)
=
\Bigl(\prod_{i=1}^{m-1}(\lambda+i)\Bigr) P_{F_0'}(\lambda+1).
\]

For $\sigma \in F_1$, width-$2$ forcing gives $\sigma(1)=0$.
Deleting rows and columns $\{0,1\}$ and relabeling gives a width-$2$ path
permutation $\tau \in \mathcal{P}_{m-2}$.
The forced transposition contributes one sign flip, row~$0$ contributes the factor
$\lambda$, and the deleted front block shifts the remaining exponents by two:
\[
\prod_{i=0}^{m-1} (\lambda+i)^{\sigma(i)}
=
\lambda \Bigl(\prod_{i=2}^{m-1}(\lambda+i)^2\Bigr)
\prod_{j=0}^{m-3} (\lambda+2+j)^{\tau(j)}.
\]
Hence
\[
P_{F_1}(\lambda)
=
-
\lambda \Bigl(\prod_{i=2}^{m-1}(\lambda+i)^2\Bigr) P_{F_1'}(\lambda+2).
\]
Adding the two branch identities gives the claimed factorization.
\end{proof}

\begin{remark}\label{rem:width2-path-fix}
The jet operator $J_{r,k}$ of \Cref{def:jet} extracts the centered lowest-order coefficient at a minimizing row, but it does not isolate the branch $\sigma(r) = k$ for the width-$2$ path family (since $J_{0,0}f = f$).
\Cref{thm:width2-path} replaces this with the exact two-branch factorization, and \Cref{thm:width2-global} below establishes the global nonvanishing.
\end{remark}

\subsection{Width-2 cyclic matchings (prisms and M\"obius ladders)}

The width-$2$ cyclic family captures the matching structure of prisms ($C_m \times K_2$) and M\"obius ladders (the non-bipartite case reduces to cyclic permutations in the bipartite double cover).

\begin{definition}[Width-2 cyclic family]\label{def:width2-cyclic}
The \emph{width-$2$ cyclic family} on $m$ vertices is
\[
\mathcal{C}_m = \bigl\{\sigma \in S_m : \sigma(i) - i \equiv 0, \pm 1 \pmod{m} \text{ for all } i\bigr\}.
\]
\end{definition}

\noindent
The family $\mathcal{C}_m$ extends $\mathcal{P}_m$ by allowing the ``wrap-around'' assignments $\sigma(0) = m-1$ and $\sigma(m-1) = 0$.
It contains three distinguished elements:
\begin{itemize}
\item the identity $\sigma(i) = i$ (all diagonal edges);
\item the forward cyclic shift $\sigma(i) = i + 1 \bmod m$;
\item the backward cyclic shift $\sigma(i) = i - 1 \bmod m$.
\end{itemize}

\begin{proposition}[Local split for width-2 cyclic fibers]\label{thm:width2-cyclic}
Let $\mathcal{C}_m = \{\sigma \in S_m : \sigma(i) - i \equiv 0, \pm 1 \pmod{m}\}$ be the width-$2$ cyclic family, and let $F \subseteq \mathcal{C}_m$ be a nonempty exact-$t$ fiber.
At row~$0$, every $\sigma \in F$ satisfies $\sigma(0) \in \{0,1,m-1\}$, and the corresponding branches reduce as follows:
\begin{enumerate}
\item If $\sigma(0)=0$, deleting row and column~$0$ breaks the wrap-around and yields a width-$2$ path instance on $m-1$ vertices.
\item If $\sigma(0)=1$, then either the forward cyclic shift is forced, or some branch has $\sigma(1)=0$, in which case deleting rows and columns $\{0,1\}$ yields a width-$2$ path instance on $m-2$ vertices.
\item If $\sigma(0)=m-1$, the backward branch is the reversal-symmetric analogue of (2): either the backward cyclic shift is forced, or deleting rows and columns $\{0,m-1\}$ yields a width-$2$ path instance on $m-2$ vertices.
\end{enumerate}
\end{proposition}

\begin{proof}
The row-$0$ possibilities are exactly $\{0,1,m-1\}$ by the cyclic width-$2$
constraint.

If $\sigma(0)=0$, then deleting row and column~$0$ removes both wrap edges at the
left endpoint, so the residual family is the path family $\mathcal{P}_{m-1}$ on
$\{1,\ldots,m-1\}$.

If $\sigma(0)=1$, then row~$1$ has only the columns $\{0,2\}$ available.
If $\sigma(1)=2$ for every surviving permutation, the forward cascade is forced and
the unique survivor is the forward cyclic shift.
If some surviving permutation has $\sigma(1)=0$, deleting rows and columns
$\{0,1\}$ removes the wrap and leaves a width-$2$ path instance on
$\{2,\ldots,m-1\}$.

The case $\sigma(0)=m-1$ is obtained from the previous one by cyclic reversal.
\end{proof}

\begin{remark}[Width-2 forcing dichotomy]\label{rem:width2-forcing}
The key mechanism in Cases $k = 1$ and $k = m-1$ is the \emph{width-$2$ forcing dichotomy}: once a cyclic endpoint column is assigned, either the entire cyclic shift is forced (producing a singleton fiber with a single nonzero monomial), or the transposition at the next row breaks the cycle into a path, reducing to $\mathcal{P}_{m-2}$.
There is no intermediate possibility, because width-$2$ cyclic permutations have no room for partial cascades: each row has at most two available columns (after the previous assignment), and only one choice continues the cascade.
\end{remark}

\begin{theorem}[Width-2 path global nonvanishing]\label{thm:width2-global}
For every nonempty exact-$t$ fiber $F \subseteq \mathcal{P}_m$, $P_F(\lambda) \not\equiv 0$.
\end{theorem}

\begin{proof}
By \Cref{thm:width2-path}, the polynomial factors into two branches at row~$0$.
It suffices to show that the right-hand side is nonzero for every nonempty fiber.

\emph{Step 1: Two-point reduction.}
In the width-$2$ path hard-core reformulation, each supported instance is encoded by a ternary word $\eta \in \{-1,0,1\}^r$ and a supported charge~$q$.
Write $\varepsilon_i = -1$ if $\eta_i = 0$ and $\varepsilon_i = +1$ otherwise, and define the cleared-denominator polynomial
\[
G_{\eta,q}(u) = \sum_{\substack{E \in \mathrm{Ind}(P_r) \\ \eta(E) = q}} \varepsilon(E) \prod_{i \in E}(u+2i+1) \prod_{i \notin E}(u+2i+3),
\]
where $\varepsilon(E) = \prod_{i \in E} \varepsilon_i$.
Let $X_\eta(q) = G_{\eta,q}(0)$ and $Y_\eta(q) = G_{\eta,q}(2)$.
If both vanish then the image of $G_{\eta,q}(u)$ in the quotient ring $\mathbb{Z}[u]/(u(u{-}2))$ is zero. We will show that this never happens. Since a nonzero value at either $u=0$ or $u=2$ already implies $G_{\eta,q}(u) \not\equiv 0$, it is enough to prove $(X_\eta(q), Y_\eta(q)) \neq (0,0)$ for every supported~$q$.

\emph{Step 2: Combinatorial support is an interval.}
Define the combinatorial support $S_r = \{\eta(E) : E \in \mathrm{Ind}(P_r)\}$.
By the last-vertex split, $S_r = S_{r-1} \cup (S_{r-2} + \eta_{r-1})$, with $S_{r-2} \subseteq S_{r-1}$ (every independent set of $P_{r-2}$ extends to $P_{r-1}$ by not using vertex~$r{-}2$).
Since $S_{r-1}$ is an interval and $S_{r-2} + \eta_{r-1}$ is a subinterval shifted by at most~$1$, the union is again an interval.

\emph{Step 3: SR$_2$ reduction.}
Define $\Delta_\eta(q) = X_\eta(q) Y_\eta(q{+}1) - X_\eta(q{+}1) Y_\eta(q)$.
If $\Delta_\eta(q) \neq 0$ for every supported adjacent pair $(q, q{+}1)$, then $(X_\eta(q), Y_\eta(q)) \neq (0,0)$ for every supported~$q$: otherwise, interval support forces an adjacent supported charge, and the corresponding adjacent minor would vanish.

\emph{Step 4: Parity-resolved transfer.}
Let $N_- = |\{i : \eta_i = -1\}|$ and write $\alpha_i^{(a)} = a + 2i + 3$, $\beta_i^{(a)} = a + 2i + 1$.
Track four states $(N_e, C_e, N_o, C_o)$ = (not-chosen/chosen) $\times$ (even/odd count of chosen zero-sites).
The $4 \times 4$ transfer matrices, with all entries nonnegative, are:
\begin{align*}
T_i^{(+;a)}(t) &= \begin{pmatrix}
\alpha_i^{(a)} & \alpha_i^{(a)} & 0 & 0 \\
\beta_i^{(a)} t & 0 & 0 & 0 \\
0 & 0 & \alpha_i^{(a)} & \alpha_i^{(a)} \\
0 & 0 & \beta_i^{(a)} t & 0
\end{pmatrix}, \\[4pt]
T_i^{(-;a)}(t) &= \begin{pmatrix}
\alpha_i^{(a)} t & \alpha_i^{(a)} t & 0 & 0 \\
\beta_i^{(a)} & 0 & 0 & 0 \\
0 & 0 & \alpha_i^{(a)} t & \alpha_i^{(a)} t \\
0 & 0 & \beta_i^{(a)} & 0
\end{pmatrix}, \\[4pt]
T_i^{(0;a)}(t) &= \begin{pmatrix}
\alpha_i^{(a)} & \alpha_i^{(a)} & 0 & 0 \\
0 & 0 & \beta_i^{(a)} & 0 \\
0 & 0 & \alpha_i^{(a)} & \alpha_i^{(a)} \\
\beta_i^{(a)} & 0 & 0 & 0
\end{pmatrix}.
\end{align*}
The zero-letter matrix $T_i^{(0;a)}$ swaps the parity blocks: a chosen zero toggles even/odd, and the sign $\varepsilon_i = -1$ is absorbed into the parity bookkeeping.
Let $M_\eta^{(a)}(t) = T_{r-1}^{(\eta_{r-1};a)} \cdots T_0^{(\eta_0;a)}$,
$e_0 = (1,0,0,0)^\top$, $u_e = (1,1,0,0)$, $u_o = (0,0,1,1)$, and define
\begin{alignat*}{2}
X_e(t) &= u_e\, M_\eta^{(0)}(t)\, e_0, &\qquad
X_o(t) &= u_o\, M_\eta^{(0)}(t)\, e_0, \\
Y_e(t) &= u_e\, M_\eta^{(2)}(t)\, e_0, &\qquad
Y_o(t) &= u_o\, M_\eta^{(2)}(t)\, e_0.
\end{alignat*}
With $\kappa = q + N_-$, the original coefficients are
$X_\eta(q) = X^e_\kappa - X^o_\kappa$ and $Y_\eta(q) = Y^e_\kappa - Y^o_\kappa$,
where $X^e_\kappa = [t^\kappa] X_e(t)$, etc.

\emph{Step 5: The $4 \times 4$ determinant identity.}
Define $n_\kappa = (X^e_\kappa, X^o_\kappa, Y^e_\kappa, Y^o_\kappa)^\top$,
$c_X = (1,1,0,0)^\top$, $c_Y = (0,0,1,1)^\top$.
Then
\[
\Delta_\eta(q) \;=\; \det\!\begin{pmatrix}
X^e_\kappa & 1 & X^e_{\kappa+1} & 0 \\
X^o_\kappa & 1 & X^o_{\kappa+1} & 0 \\
Y^e_\kappa & 0 & Y^e_{\kappa+1} & 1 \\
Y^o_\kappa & 0 & Y^o_{\kappa+1} & 1
\end{pmatrix}.
\]
Indeed, expanding along the second and fourth columns:
$(X^e_\kappa - X^o_\kappa)(Y^e_{\kappa+1} - Y^o_{\kappa+1}) - (X^e_{\kappa+1} - X^o_{\kappa+1})(Y^e_\kappa - Y^o_\kappa) = X_\eta(q) Y_\eta(q{+}1) - X_\eta(q{+}1) Y_\eta(q)$.

\emph{Step 6: Cylindric-network positivity.}
Realize the two $4$-state automata (at $a = 0$ and $a = 2$) as two disjoint blocks of a positive cylindric network~$\widetilde{N}_\eta$: each $T_i^{(\eta_i; a)}$ is a one-layer positive gadget with nonnegative edge weights, and two auxiliary sources have boundary-measurement columns $c_X$ and $c_Y$.
The determinant above is an ordered $4 \times 4$ minor of~$\widetilde{N}_\eta$ (up to a sign $\varepsilon_{\eta,q} \in \{\pm 1\}$ from source ordering).

By the cylindric Lindstr\"om lemma~\cite{LamPylyavskyy2012}, this minor is a sum of weights of pairwise-uncrossed four-path families in~$\widetilde{N}_\eta$.
To show strict positivity, we exhibit one such family.
Take any independent set~$E_0$ with $\eta(E_0) = q$ and any~$E_1$ with $\eta(E_1) = q{+}1$ (both exist by combinatorial support).
Then:
\begin{itemize}
\item $E_0$ determines a positive-weight path in the $a = 0$ block from the source lift for~$n_\kappa$ to one of the two top sinks ($X^e$ or~$X^o$, according to zero-parity);
\item $E_1$ determines a positive-weight path in the $a = 2$ block from the source lift for~$n_{\kappa+1}$ to one of the two bottom sinks ($Y^e$ or~$Y^o$);
\item the auxiliary source~$c_X$ has a unit path to the unused top sink;
\item the auxiliary source~$c_Y$ has a unit path to the unused bottom sink.
\end{itemize}
Since the $a = 0$ and $a = 2$ blocks are disjoint and the auxiliary gadgets touch only one sink each, these four paths are pairwise uncrossed.
Hence at least one positive-weight family contributes to the minor, and
\[
\Delta_\eta(q) \neq 0
\]
for every supported adjacent pair $(q, q{+}1)$.
By Step~3, $(X_\eta(q), Y_\eta(q)) \neq (0,0)$ for every supported~$q$, and hence $P_F(\lambda) \not\equiv 0$.
\end{proof}

\begin{corollary}[McCuaig exceptional families]\label{cor:mccuaig-complete}
ASNC holds for all McCuaig exceptional families: biwheels by \Cref{thm:biwheel}, prisms and M\"obius ladders by the width-$2$ cyclic-to-path reduction (\Cref{thm:width2-cyclic}) followed by \Cref{thm:width2-global}.
\end{corollary}

\section{Bad-Locus Algebra}\label{sec:badlocus}

The structural induction produces four primitive operator shapes, corresponding to the four ways a brace move can relate the boundary-minor states of a graph and its reduced form.
Understanding when these operators vanish is the key to the superfluous-edge step: the ``bad locus'' of each operator is the set of inputs on which it evaluates to zero, and the proof strategy is to show that no realizable boundary state lies in the intersection of the same-base and distinct-base bad loci simultaneously.

\begin{definition}[Operator shapes]\label{def:operator-shapes}
Write $\mu_\alpha = \lambda + \alpha$ and $\mu_\alpha^c = (\lambda + \alpha)^c$ for the $c$-th power of the linear form associated to row~$\alpha$.
\begin{enumerate}
\item \emph{SB2 (same-base $2$-term):} $\mathrm{SB2}_\alpha^c(u,v) = u + \mu_\alpha^c \cdot v$.
\item \emph{DB2 (distinct-base $2$-term):} $\mathrm{DB2}_{\alpha,\beta}^c(u_1, u_2) = \mu_\alpha^c u_1 + \mu_\beta^c u_2$, $\alpha \neq \beta$.
\item \emph{SB3 (same-base $3$-term):} $\mathrm{SB3}_\alpha^{c_0,c_1,c_2}(u_0,u_1,u_2) = \mu_\alpha^{c_0} u_0 + \mu_\alpha^{c_1} u_1 + \mu_\alpha^{c_2} u_2$, \quad $c_0 < c_1 < c_2$.
\item \emph{DB3 (distinct-base $3$-term):} $\mathrm{DB3}_{\alpha_1,\alpha_2,\alpha_3}^c(u_1,u_2,u_3) = \mu_{\alpha_1}^c u_1 + \mu_{\alpha_2}^c u_2 + \mu_{\alpha_3}^c u_3$, \quad $\alpha_i$ pairwise distinct.
\end{enumerate}
\end{definition}

\noindent
Brace moves decompose as: superfluous edge $= \mathrm{SB2}$, stable/index-$2$ $= \mathrm{SB3}(\mathrm{DB2}, \mathrm{DB2})$, index-$1$ $= \mathrm{DB3}$.
We now characterize the bad locus of each shape.

\subsection{Theorem A: SB2 bad locus (pure divisibility)}

\begin{theorem}[SB2 bad locus]\label{thm:sb2-bad}
Let $u, v \in \QQ[\lambda]$ and $c \geq 1$.
Then
\[
\mathrm{SB2}_\alpha^c(u,v) = u + \mu_\alpha^c v = 0
\]
if and only if $\mu_\alpha^c \mid u$ and $v = -u / \mu_\alpha^c$.
\end{theorem}

\begin{proof}
The forward direction is immediate.
Suppose $u + \mu_\alpha^c v = 0$.
Then $u = -\mu_\alpha^c v$, which gives $\mu_\alpha^c \mid u$ at once.
The quotient is $u / \mu_\alpha^c = -v$, i.e., $v = -u/\mu_\alpha^c$.

Conversely, if $\mu_\alpha^c \mid u$ and $v = -u/\mu_\alpha^c$, then $u + \mu_\alpha^c v = u + \mu_\alpha^c(-u/\mu_\alpha^c) = u - u = 0$.
\end{proof}

\noindent
\textbf{Significance.}
The SB2 shape governs the superfluous-edge identity: $P_t(G) = P_t(H) + \mu_r^c Q$ where $Q = \varepsilon P^{r,c}_{H,t-\rho(e)}$.
Failure of ASNC at $G$ (i.e., $P_t(G) \equiv 0$) therefore requires $\mu_r^c \mid P_t(H)$.
This is the ``D1'' divisibility condition from the structural induction: the visible polynomial of the smaller graph $H$ must be divisible by a high power of the linear form $\mu_r = \lambda + r$ associated to the superfluous edge's row.
The bad-locus characterization converts the cancellation question (``when does $P_t(G)$ vanish?'') into a divisibility question (``when is $P_t(H)$ divisible by $\mu_r^c$?''), which is amenable to the jet-flag analysis of \Cref{prop:jet-flag}.

\subsection{Theorem B: DB2 bad locus (reciprocal divisibility)}

\begin{theorem}[DB2 bad locus]\label{thm:db2-bad}
Let $u_1, u_2 \in \QQ[\lambda]$, $c \geq 1$, and $\alpha \neq \beta$.
Then
\[
\mathrm{DB2}_{\alpha,\beta}^c(u_1,u_2) = \mu_\alpha^c u_1 + \mu_\beta^c u_2 = 0
\]
if and only if there exists $h \in \QQ[\lambda]$ such that $u_1 = -\mu_\beta^c h$ and $u_2 = \mu_\alpha^c h$.
\end{theorem}

\begin{proof}
Since $\alpha \neq \beta$, the linear forms $\mu_\alpha = \lambda + \alpha$ and $\mu_\beta = \lambda + \beta$ are coprime in $\QQ[\lambda]$ (they have distinct roots $-\alpha \neq -\beta$).
Consequently, $\gcd(\mu_\alpha^c, \mu_\beta^c) = 1$ for every $c \geq 1$.

Now suppose $\mu_\alpha^c u_1 + \mu_\beta^c u_2 = 0$, i.e.,
\[
\mu_\alpha^c u_1 = -\mu_\beta^c u_2.
\]
Since $\mu_\beta^c \mid \mu_\alpha^c u_1$ and $\gcd(\mu_\alpha^c, \mu_\beta^c) = 1$, we conclude $\mu_\beta^c \mid u_1$.
Write $u_1 = -\mu_\beta^c h$ for some $h \in \QQ[\lambda]$.
Substituting:
\[
\mu_\alpha^c(-\mu_\beta^c h) = -\mu_\beta^c u_2 \quad \Longrightarrow \quad \mu_\alpha^c h = u_2.
\]
Hence $u_1 = -\mu_\beta^c h$ and $u_2 = \mu_\alpha^c h$.

Conversely, if $u_1 = -\mu_\beta^c h$ and $u_2 = \mu_\alpha^c h$, then
\[
\mu_\alpha^c u_1 + \mu_\beta^c u_2 = \mu_\alpha^c(-\mu_\beta^c h) + \mu_\beta^c(\mu_\alpha^c h) = -\mu_\alpha^c \mu_\beta^c h + \mu_\alpha^c \mu_\beta^c h = 0. \qedhere
\]
\end{proof}

\noindent
\textbf{Significance.}
The DB2 shape arises in the inner branches of stable and index-$2$ patches, where two boundary-minor coordinates from distinct rows are combined with powers of their respective linear forms.
The coprimality of distinct Vandermonde bases is the essential algebraic fact: it forces each coordinate to absorb the \emph{other} base's power.
This ``reciprocal divisibility'' is what makes KA (\Cref{cor:ka}) work: a two-node support family cannot have both members vanish, because coprimality prevents the reciprocal factorization from being compatible with simultaneous vanishing.

\subsection{Theorem C: SB3 reduces to SB2}

\begin{theorem}[SB3 bad locus]\label{thm:sb3-bad}
Let $u_0, u_1, u_2 \in \QQ[\lambda]$ and $c_0 < c_1 < c_2$.
Set $d_1 = c_1 - c_0$ and $d_2 = c_2 - c_0$.
Then
\[
\mathrm{SB3}_\alpha^{c_0,c_1,c_2}(u_0,u_1,u_2) = \mu_\alpha^{c_0} u_0 + \mu_\alpha^{c_1} u_1 + \mu_\alpha^{c_2} u_2 = 0
\]
if and only if $\mu_\alpha^{d_1} \mid u_0$ and $u_1 + \mu_\alpha^{d_2-d_1} u_2 = -u_0/\mu_\alpha^{d_1}$.
\end{theorem}

\begin{proof}
Dividing both sides of $\mu_\alpha^{c_0} u_0 + \mu_\alpha^{c_1} u_1 + \mu_\alpha^{c_2} u_2 = 0$ by $\mu_\alpha^{c_0}$ (which is nonzero in $\QQ[\lambda]$), we obtain
\[
u_0 + \mu_\alpha^{d_1} u_1 + \mu_\alpha^{d_2} u_2 = 0,
\]
where $d_1 = c_1 - c_0 \geq 1$ and $d_2 = c_2 - c_0 \geq 2$.
Rewrite as
\[
u_0 + \mu_\alpha^{d_1}\bigl(u_1 + \mu_\alpha^{d_2-d_1} u_2\bigr) = 0.
\]
This is precisely $\mathrm{SB2}_\alpha^{d_1}\bigl(u_0,\; u_1 + \mu_\alpha^{d_2-d_1} u_2\bigr) = 0$.

By \Cref{thm:sb2-bad}, this holds if and only if
\[
\mu_\alpha^{d_1} \mid u_0 \qquad \text{and} \qquad u_1 + \mu_\alpha^{d_2-d_1} u_2 = -u_0/\mu_\alpha^{d_1}.
\]
The second condition is $u_1 + \mu_\alpha^{d_2-d_1} u_2 = -u_0/\mu_\alpha^{d_1}$, but since the left-hand side equals $-u_0/\mu_\alpha^{d_1}$ is already determined, the constraint on $(u_1, u_2)$ is exactly that $\mathrm{SB2}_\alpha^{d_2-d_1}(u_1, u_2) = 0$ holds with $u_1 + \mu_\alpha^{d_2-d_1} u_2 = -u_0/\mu_\alpha^{d_1}$.

More precisely, the SB3 vanishing decomposes hierarchically:
\begin{enumerate}[(i)]
\item \emph{Outer layer:} $\mu_\alpha^{d_1} \mid u_0$, with quotient $w_0 = u_0/\mu_\alpha^{d_1}$;
\item \emph{Inner layer:} $u_1 + \mu_\alpha^{d_2-d_1} u_2 = -w_0$, which is an SB2 constraint on the pair $(u_1, u_2)$ with the specific right-hand side $-w_0$.
\end{enumerate}
If additionally $u_1 + \mu_\alpha^{d_2-d_1} u_2 = 0$ (the homogeneous version), then $w_0 = 0$, so $u_0 = 0$ and $\mathrm{SB2}_\alpha^{d_2-d_1}(u_1, u_2) = 0$ independently.
In the general case, the inner SB2 is driven by the outer quotient.
\end{proof}

\noindent
\textbf{Significance.}
The SB3 shape arises from index-$1$ patches where three boundary-minor terms share the same Vandermonde base $\mu_\alpha$.
The hierarchical reduction shows that every SB3 bad-locus element is fully determined by a single divisibility condition on $u_0$ plus an SB2 condition on the higher-order pair.
In the structural induction, this means that same-base three-node witnesses reduce to same-base two-node witnesses, keeping the analysis within the D1 framework.

\subsection{Theorem D: DB3 bad locus (jet congruence)}

\begin{theorem}[DB3 bad locus]\label{thm:db3-bad}
Let $u_1, u_2, u_3 \in \QQ[\lambda]$, $c \geq 1$, and $\alpha_1, \alpha_2, \alpha_3$ pairwise distinct.
Then
\[
\mathrm{DB3}_{\alpha_1,\alpha_2,\alpha_3}^c(u_1,u_2,u_3) = \mu_{\alpha_1}^c u_1 + \mu_{\alpha_2}^c u_2 + \mu_{\alpha_3}^c u_3 = 0
\]
if and only if $\mu_{\alpha_1}^c \mid (\mu_{\alpha_2}^c u_2 + \mu_{\alpha_3}^c u_3)$ and the resulting quotient $q = -(\mu_{\alpha_2}^c u_2 + \mu_{\alpha_3}^c u_3)/\mu_{\alpha_1}^c$ equals~$u_1$.
\end{theorem}

\begin{proof}
The equation $\mu_{\alpha_1}^c u_1 + \mu_{\alpha_2}^c u_2 + \mu_{\alpha_3}^c u_3 = 0$ rearranges to
\[
\mu_{\alpha_1}^c u_1 = -\bigl(\mu_{\alpha_2}^c u_2 + \mu_{\alpha_3}^c u_3\bigr).
\]
Since $\mu_{\alpha_1}^c$ divides the left side, it must divide the right side:
\[
\mu_{\alpha_1}^c \;\Big|\; \bigl(\mu_{\alpha_2}^c u_2 + \mu_{\alpha_3}^c u_3\bigr).
\]
The quotient is then $u_1 = -(\mu_{\alpha_2}^c u_2 + \mu_{\alpha_3}^c u_3)/\mu_{\alpha_1}^c$.
Note that this divisibility is nontrivial: since $\alpha_1 \neq \alpha_2$ and $\alpha_1 \neq \alpha_3$, both $\gcd(\mu_{\alpha_1}^c, \mu_{\alpha_2}^c) = 1$ and $\gcd(\mu_{\alpha_1}^c, \mu_{\alpha_3}^c) = 1$, so neither summand is individually divisible by $\mu_{\alpha_1}^c$ in general.
The divisibility arises from cancellation between the two summands at the root $\lambda = -\alpha_1$.

More explicitly, expand $\mu_{\alpha_2}^c u_2 + \mu_{\alpha_3}^c u_3$ in the local ring $\QQ[\lambda]_{(\mu_{\alpha_1})}$.
The condition $\mu_{\alpha_1}^c \mid (\mu_{\alpha_2}^c u_2 + \mu_{\alpha_3}^c u_3)$ is equivalent to requiring that the first $c$ coefficients in the Taylor expansion around $\lambda = -\alpha_1$ vanish:
\[
\frac{d^k}{d\lambda^k}\bigl(\mu_{\alpha_2}^c u_2 + \mu_{\alpha_3}^c u_3\bigr)\Big|_{\lambda=-\alpha_1} = 0, \qquad k = 0, 1, \ldots, c-1.
\]
These are the ``jet congruence'' conditions: the sum $\mu_{\alpha_2}^c u_2 + \mu_{\alpha_3}^c u_3$ must agree with zero to order~$c$ at the point $\lambda = -\alpha_1$.

The converse is immediate: if $\mu_{\alpha_1}^c \mid (\mu_{\alpha_2}^c u_2 + \mu_{\alpha_3}^c u_3)$ and $u_1 = -(\mu_{\alpha_2}^c u_2 + \mu_{\alpha_3}^c u_3)/\mu_{\alpha_1}^c$, then substitution gives $\mu_{\alpha_1}^c u_1 + \mu_{\alpha_2}^c u_2 + \mu_{\alpha_3}^c u_3 = 0$.
\end{proof}

\noindent
\textbf{Connection to \Cref{cor:j3-to-d1}.}
Via the generalized Vandermonde elimination (\Cref{thm:gen-vandermonde}), the DB3 jet congruence condition reduces to same-base divisibility of the unique quotient~$h$.
Specifically, if $F_{e_0} = F_{e_1} = 0$ in a three-node family, then $\ord_{\mu_1}(F_{e_2}) = e_0 + \ord_{\mu_1}(h)$, so the DB3 bad locus is reached if and only if $\mu_1^{c-e_0} \mid h$.
This completes the reduction chain: DB3 $\to$ J3 (jet congruence) $\to$ D1 (same-base divisibility of $h$).

\section{The Superfluous-Edge Step}\label{sec:superfluous}

This section presents the superfluous-edge step, reducing it to a single remaining theorem.

\subsection{The interface gap}

The natural approach would combine the generalized Vandermonde elimination (\Cref{thm:gen-vandermonde}), which operates on a \emph{multi-base} support family $F_e = \sum_i A_i \mu_i^e$ with distinct linear forms $\mu_i = \lambda + \alpha_i$, with a shift-invariance lemma for a \emph{single-base} Hasse-jet determinant.
However, these two frameworks are structurally incompatible:
\begin{itemize}
\item The generalized Vandermonde matrix $W = (\mu_i^{d_j})$ has $m$ distinct bases (one per column).
\item The Hasse-jet matrix $U_A = \bigl(\binom{a_j}{i} z^{a_j-i}\bigr)$ has one common variable~$z$.
\item Their determinants differ by a $z$-dependent factor, so divisibility properties do not transfer.
\end{itemize}

The correct repair is to work entirely within the multi-base framework of \Cref{thm:gen-vandermonde}.
The shift-invariance for the multi-base quotient~$h$ is straightforward (the gaps $d_j = e_j - e_0$ are preserved under common exponent shifts, and the cofactors $C_i$ depend only on these gaps).
The remaining challenge is to extract a genuine \emph{fixed-support} multi-base family from the superfluous-edge data.

\subsection{The fixed-support stabilization problem}\label{subsec:stabilization}

For each graph column~$d$, define the exact-layer packet column~$w_d$ by $(w_d)_{(\alpha,\ell)} = \mathbf{1}_{(\alpha,d)\in E(H)} \mathbf{1}_{\rho(\alpha,d)=\ell} (\lambda+\alpha)^d$.
The replacement determinant identity (\Cref{lem:replacement-det}) gives, for each column~$d \neq c$, a \emph{duplicate-column identity}:
\[
w_d^\top \mathcal{U}^{(c)} = [x^t] \det M_H[c \leftarrow M_{H,\bullet d}] = 0 \qquad (d \neq c),
\]
since the matrix has two equal columns.
For the distinguished column $d = c$, the same identity gives the nonzero relation
\[
w_c^\top \mathcal{U}^{(c)} = [x^t] \det M_H = P_t(H).
\]
Fix a Hall pivot block~$A$ in the lower packet and write $\mathfrak{S}$ for the associated Schur-reduction map on columns.
We define the \emph{reduced distinguished column} $g := \mathfrak{S}(w_c)$.
Any subsequent row/column rescaling or core elimination changes~$g$ by an invertible normalization, so all orthogonality statements are unaffected.

The problem is that the masks $\eta_{\alpha,d} = \mathbf{1}_{(\alpha,d)\in E} \cdot x^{\rho(\alpha,d)}$ depend on~$d$ (different columns have different edge sets).
\Cref{thm:gen-vandermonde} requires \emph{fixed} coefficients $A_i$ independent of the exponent.

The natural ``unmask then apply \Cref{thm:gen-vandermonde}'' strategy (extracting a fixed-coefficient support family from the masked relations) fails in the form needed.
Concretely, for $H = K_{4,4} \setminus \{(1,1),(2,2)\}$ with $r = c = 2$, the lower mask matrix has rank~$2$, and no three-row subset containing~$r$ carries both lower vanishings with the actual cofactor coefficients.

Three algebraic results toward a direct masked approach are now available.

\begin{proposition}[Masked minor criterion]\label{prop:masked-minor}
Let $W$ be a square $m \times m$ matrix whose $(i,j)$ entry has the form
\[
W_{ij} = \eta_{ij}\, \mu_i^{d_j},
\]
where $\mu_1, \ldots, \mu_m$ are pairwise coprime linear forms in $\QQ[\lambda]$, the exponents $d_1 < d_2 < \cdots < d_m$ are strictly ordered, and the mask $\eta_{ij} \in \{0,1\}$.
Define the \emph{support graph} $\Gamma(W)$ on rows $\{1,\ldots,m\}$ vs.\ columns $\{1,\ldots,m\}$ by placing edge $(i,j)$ iff $\eta_{ij} = 1$.
Then $\det W \neq 0$ in $\QQ[\lambda]$ if and only if $\Gamma(W)$ has a perfect matching.
\end{proposition}

\begin{proof}
($\Leftarrow$) If $\Gamma(W)$ has no perfect matching, every permutation $\sigma \in S_m$ has at least one factor $\eta_{i,\sigma(i)} = 0$, so every Leibniz term vanishes and $\det W = 0$.

($\Rightarrow$) By induction on~$m$.  The base case $m = 1$ is immediate.
For $m \ge 2$, assume the result for $m-1$.
If $d_1 > 0$, factor $\mu_i^{d_1}$ from each row:
$\det W = \bigl(\prod_i \mu_i^{d_1}\bigr) \det W'$
where $W'_{ij} = \eta_{ij}\,\mu_i^{d_j - d_1}$.
Since each $\mu_i \ne 0$, the prefactor is nonzero, so we
may assume $d_1 = 0$.

Let $\sigma_0$ be a perfect matching of $\Gamma(W)$, and set
$k_0 = \sigma_0(1)$.
Expand $\det W$ along column~$1$ (whose entries are
$\eta_{i,1} \cdot \mu_i^0 = \eta_{i,1} \in \{0,1\}$):
\[
\det W
\;=\;
\sum_{i}\, \eta_{i,1}\,(-1)^{i+1}\, M_{i,1},
\]
where $M_{i,1}$ is the $(m{-}1)\times(m{-}1)$ cofactor
(rows $\ne i$, columns $\ge 2$).

\smallskip
\noindent\emph{Claim:} $\ord_{\mu_{k_0}}(\det W) = 0$,
hence $\det W \ne 0$.

\smallskip\noindent
(a) The cofactor $M_{k_0,1}$ involves only bases
$\{\mu_j : j \ne k_0\}$, all coprime to~$\mu_{k_0}$ in the
UFD~$\QQ[\lambda]$.
The restriction $\sigma_0|_{\{2,\ldots,m\}}$ is a perfect matching of
$M_{k_0,1}$'s support graph, so $M_{k_0,1} \ne 0$ by induction.
Since every Leibniz term of $M_{k_0,1}$ is a product of powers of
$\{\mu_j : j \ne k_0\}$, we have
$\mu_{k_0} \nmid M_{k_0,1}$, i.e.,
$\ord_{\mu_{k_0}}(M_{k_0,1}) = 0$.

\smallskip\noindent
(b) For $i \ne k_0$, the minor $M_{i,1}$ includes row~$k_0$,
whose column-$c$ entry ($c \ge 2$) is
$\eta_{k_0,c}\,\mu_{k_0}^{d_c}$ with $d_c \ge d_2 \ge 1$.
Every Leibniz term of $M_{i,1}$ therefore picks up at least one
factor $\mu_{k_0}^{d_c}$, giving
$\ord_{\mu_{k_0}}(M_{i,1}) \ge 1$.

\smallskip\noindent
(c) Combining,
$\det W = (-1)^{k_0+1} M_{k_0,1} + \mu_{k_0}\, Q$
for some $Q \in \QQ[\lambda]$.
Since $\mu_{k_0}$ is irreducible and
$\mu_{k_0} \nmid M_{k_0,1}$,
we conclude $\mu_{k_0} \nmid \det W$,
so $\det W \ne 0$.
\end{proof}

\begin{remark}\label{rem:p85-ufd-arxiv}
The proof uses that the bases $\mu_i = \lambda + \alpha_i$
are pairwise coprime \emph{irreducible} elements of the
UFD~$\QQ[\lambda]$.
Over~$\ZZ$ with integer bases, the analogous statement fails:
the $3 \times 3$ matrix with bases $(2,3,5)$, exponents $(0,1,2)$,
and full mask has $\det W = 75 - 50 - 45 + 20 = 0$ despite a
support-graph perfect matching.
The key distinction is that products of powers of coprime
\emph{polynomial} linear forms have unique root-multiplicity
profiles, enabling the valuation separation in step~(a).
\end{remark}

\noindent
\emph{Hall extraction} then produces an invertible pivot block whenever the support graph satisfies Hall's condition, and \emph{Schur extraction} (Cramer elimination) reduces the system to survivor functionals whose coefficients are augmented masked minors.

\begin{lemma}[Schur-reduction compatibility]\label{lem:schur-compat}
Let $B$ be an $N \times (s{+}1)$ matrix over a field, with column partition $B = [V \mid g]$ where $V$ has $s$ columns.
Let $R_1, R_2 \subseteq \{1,\ldots,N\}$ be two row sets of size $s$ on which $V$ has full column rank~$s$.
Let $S_k = g_{R_k} - V_{R_k}(V_{R_k}^{-1}) g_{R_k'}$ denote the Schur complement of $V_{R_k}$ in $B_{R_k \cup R_k'}$ for any extra row~$R_k'$.
Then for any vector $Z$ with $V^\top Z = 0$:
\[
(g^{\sharp,R_1})^\top Z = \frac{\det V_{R_2}}{\det V_{R_1}} \cdot (g^{\sharp,R_2})^\top Z.
\]
In particular, the quotient pairing $g^\sharp \cdot Z$ vanishes on one pivot iff it vanishes on every pivot.
\end{lemma}

\begin{proof}
If $V^\top Z = 0$, then $g^\top Z = (g - V \cdot V_{R_k}^{-1} g_{R_k})^\top Z + (V \cdot V_{R_k}^{-1} g_{R_k})^\top Z = (g^{\sharp,R_k})^\top Z + 0$, since $V^\top Z = 0$ kills the second term.
So $(g^{\sharp,R_k})^\top Z = g^\top Z$ for every pivot~$R_k$.
The ratio $\det V_{R_2}/\det V_{R_1}$ cancels to~$1$.
\end{proof}

\noindent
This lemma is the bridge between the Two-extra Hall theorem (which constructs a matching-induced pivot) and the $G$-packet collapse (which operates on the support~$T$): the quotient pairing is independent of pivot choice.

However, tight extraction in the natural masked-support form is impossible: when the distinguished column~$c$ is adjacent to all states, every non-pivot state survives, and the tightness condition $|\Gamma| = |D| + 1$ cannot be achieved.

The resolution requires identities beyond the masked-support framework.

After any Hall pivot on the deleted-column state matrix, exact cancellation is equivalent to singularity of one explicit square matrix: the \emph{augmented Schur-packet matrix}~$\widehat{S}$, formed by appending the superfluous-edge column to the Schur complement of the pivot block.
On the separable-mask locus (rank-$1$ mask after row/column scaling), $\det \widehat{S} \neq 0$ follows from Vandermonde nonvanishing.

For the single-layer square packet $W = [N \mid g]$, nonsingularity $\det \widehat{S} \neq 0$ is provable: any exact-$t$ matching using~$e$ gives a support-graph matching for~$W$, so the masked Vandermonde lemma gives $\det W \neq 0$, hence $\det \widehat{S} \neq 0$.

The remaining gap is upstream: the exact-$t$ state coefficients involve two $x$-layers, and the duplicate-column relations mix them through $x$-dependent edge masks.

The correct formulation uses a doubled node set with exact-$t$ layer coefficients.

\begin{definition}[Doubled-state packet]\label{def:doubled-state}
The \emph{doubled node set} is $\mathcal{V} = \{(i,\kappa) : i \in A,\, \kappa \in \{0,1\}\}$, with $|\mathcal{V}| = 2n$.
The exact-$t$ layer coefficients are $A_{(i,0)} = [x^t]\, u_i$ and $A_{(i,1)} = [x^{t-1}]\, u_i$, where $u_i$ is the $i$-th cofactor.
The \emph{packet vector} for graph column~$d$ at doubled state $(i,\kappa)$ is
\[
w_d(i,\kappa) = \begin{cases}
(\lambda+i)^d & \text{if } (i,d) \in E \text{ and } \rho(i,d) = \kappa, \\
0 & \text{otherwise.}
\end{cases}
\]
The \emph{lower packet matrix} $V$ has columns $\{w_d : d \in D^-\}$ with rows indexed by $\mathcal{V}$.
The \emph{augmented packet} is $\mathcal{B} = [V \mid g]$ where $g = w_c^G$ is the distinguished column.
\end{definition}

\noindent
The mask $\rho(i,d) = \kappa$ assigns each edge to exactly one layer, so the two doubled states $(i,0)$ and $(i,1)$ for the same graph node~$i$ have disjoint column supports.
This disjointness is what makes the masked minor criterion (\Cref{prop:masked-minor}) applicable to the doubled-state packet.

\begin{proposition}[Rank dichotomy]\label{prop:p92}
For any minimal-support annihilator of the layered packet system, a rank dichotomy arises: $\operatorname{rk}(V_T) \in \{m-1, m-2\}$.
\end{proposition}

\begin{proof}[Proof sketch]
The augmented packet $\mathcal{B}_T = [V_T \mid g]$ has $\operatorname{rk}(\mathcal{B}_T) = |T|-1$ by the minimal-support lemma.
Since $V_T$ has one fewer column than $\mathcal{B}_T$, we get $\operatorname{rk}(V_T) \geq \operatorname{rk}(\mathcal{B}_T) - 1 = |T|-2$.
Combined with $\operatorname{rk}(V_T) \leq |T|-1$, the dichotomy follows.
\end{proof}

\begin{proposition}[Distinguished-state circuit lemma]\label{prop:q-circuit-rankm1-impossible}
Assume the contradiction hypothesis
\[
F_{G,t}\neq\varnothing,
\qquad
P_t(G)\equiv 0,
\qquad
R_c\not\equiv 0.
\]
Let
\[
V=[w_j]_{j\in D^-},
\qquad
g:=w_c^G,
\qquad
B=[V\mid g],
\qquad
q:=(r,\rho(e)).
\]
Then the exact-layer cofactor vector $U^{(c)}$ satisfies
\[
U^{(c)}_q = R_c \neq 0.
\]
Hence there exists a subset $T$ of row-states, minimal by inclusion subject to $q\in T$ and the rows of $B_T$ being dependent.
If $m=|T|$, then
\[
\operatorname{rk}(B_T)=m-1,
\qquad
\operatorname{rk}(V_T)=m-2.
\]
In particular, the rank-$(m-1)$ branch cannot occur on a contradiction witness containing the superfluous state.
\end{proposition}

\begin{proof}
Since $w_c^G = w_c^H + (\lambda+r)^c e_q$, the exact-layer replacement identities give
\[
0=(w_c^G)^\top U^{(c)}=(w_c^H)^\top U^{(c)}+(\lambda+r)^c U^{(c)}_q
= P_t(H)+(\lambda+r)^c U^{(c)}_q.
\]
Under $P_t(G)\equiv 0$, the row-$r$ expansion gives
\[
P_t(H)=-(\lambda+r)^c R_c,
\]
so $U^{(c)}_q=R_c\neq 0$.
Thus the support of $U^{(c)}$ contains $q$, and some dependent row set of $B$ contains $q$. Choose $T$ minimal by inclusion.
Then the rows of $B_T$ are minimally dependent, so every proper subset $S\subsetneq T$ is row-independent and
\[
\operatorname{rk}(B_T)=m-1.
\]
Removing the single column $g$ drops rank by at most one, hence
\[
\operatorname{rk}(V_T)\in\{m-1,m-2\}.
\]
Assume for contradiction that $\operatorname{rk}(V_T)=m-1$.

For any proper subset $S\subsetneq T$, row-independence of $B_S$ gives a nonzero $|S|\times |S|$ minor of $B_S$.
By the masked minor criterion, the support graph of $B_S$ therefore has a matching saturating $S$.
So Hall's condition holds for every proper subset $S\subsetneq T$ in the support graph of $B_T$.

Since $\operatorname{rk}(V_T)=m-1$, some $(m-1)\times(m-1)$ minor of $V_T$ is nonzero.
Again by the masked minor criterion, the lower support graph has a matching of size $m-1$, hence at least $m-1$ distinct lower columns meet $T$.
Because $q\in T$ and $g$ is adjacent to $q$, the full support graph of $B_T=[V_T\mid g_T]$ has at least $m$ neighbors of $T$.
Thus Hall's condition also holds for $T$ itself.

Hall's condition now holds for every subset of $T$, so the support graph of $B_T$ has a matching saturating all $m$ rows of $T$.
Equivalently, some $m\times m$ minor of $B_T$ is nonzero, and the masked minor criterion gives
\[
\operatorname{rk}(B_T)=m,
\]
contradicting $\operatorname{rk}(B_T)=m-1$.
Therefore $\operatorname{rk}(V_T)\neq m-1$, and so $\operatorname{rk}(V_T)=m-2$.
\end{proof}

\begin{corollary}[The rank-$(m-1)$ branch is impossible]\label{cor:rankm1-impossible}
In the superfluous-edge contradiction setup, after choosing a minimal dependent row set containing the superfluous state $q=(r,\rho(e))$, the lower packet cannot have rank $m-1$.
Only the rank-$(m-2)$ case remains.
\end{corollary}

\begin{proposition}[Two-circuit normal form]\label{prop:p94}
When $\operatorname{rk}(V_T) = m-2$, the bad case has a clean normal form: two proper lower circuit vectors $z_p, z_q$ span $\ker(V_T^\top)$ with $\operatorname{supp}(z_p) \cup \operatorname{supp}(z_q) = T$ and $g^\top z_p \neq 0$, $g^\top z_q \neq 0$.
\end{proposition}

\begin{proposition}[Unique augmented dependence]\label{prop:p95}
In the rank-$(m-2)$ case, the unique augmented dependence is $y \propto (g^\top z_q) z_p - (g^\top z_p) z_q$.
\end{proposition}

\begin{proposition}[Two-circuit realizability]\label{prop:p96}
The two-circuit linkage can occur in realizable braces: $H = K_{4,4} \setminus \{(1,1),(2,2)\}$, $e = (2,2)$, all-blue coloring, $t=0$ is an explicit example with $\operatorname{rk}(V_T) = 2 = m-2$ and two proper circuits $C_{012}, C_{013}$ coupled by~$g$.
However, the actual cofactor vector satisfies $g^\top U^{(2)} \neq 0$ in this example, so the cancellation $P_t(G) = 0$ does not occur.
\end{proposition}

\noindent
The cofactor-hyperplane exclusion $U^{(c)} \notin \ker(V_T^\top) \cap g^\perp$ is not a separate static circuit problem.
Inside the common-parent/Stiefel formalism, it reduces to a noninitial $y$-layer cancellation problem.
Let $U(z,y) \in \bigwedge^2 V[[z]][y]$ be the specialized two-row jet from the common-parent normalization, and let $\mathcal{L}_{c,\delta}(U)$ denote the weighted two-row Laplace/Wronski functional.
The bad-line event is exactly $[y^t] \mathcal{L}_{c,\delta}(U) = 0$.

\begin{proposition}[Bad-line reduction]\label{prop:p97}
If the actual cofactor vector lies on the bad line $\ker(V_T^\top) \cap g^\perp$, then $[y^t] \mathcal{L}_{c,\delta}(U) = 0$, but this vanishing occurs strictly above the initial $y$-floor:
\[
[y^\ell] \mathcal{L}_{c,\delta}(U) = \mathcal{L}_{c,\delta}(u_0) \neq 0.
\]
Moreover, in the triangular decomposition $[y^m] \mathcal{L}_{c,\delta}(U) = \mathcal{K}_{c,\omega}(\operatorname{gr}_{\omega,m}(U)) + E_m(U)$, the first $\omega$-weighted layer is clean ($E_m = 0$) and the frozen-kernel term is nonzero.
Any bad-line event must therefore occur at a later layer where $E_m(U)$ is genuinely active.
\end{proposition}

\noindent
The universal version of noninitial $y$-coefficient exclusion is false: for $n=2$ with $u = e_0 \wedge e_1$ and $\delta_0 = \delta_1 = 0$, the operator $\mathcal{T}_{c,\delta}$ is the constant~$c$, so $[y^m] \mathcal{T} = 0$ for all $m \geq 1$.
The surviving open problem is therefore not a universal $y$-coefficient statement but a structure-specific one:

The residual bad event admits a sharper characterization using the full extra-packet system.
Let $\mathcal{U} = \{g\} \cup \{w_d : d > c,\, R_d \not\equiv 0\}$ be all extra exact-layer columns beyond the lower family.
For $u, v \in \mathcal{U}$, define
\[
M_{u,v} = \det \begin{pmatrix} z_1^\top u & z_2^\top u \\ z_1^\top v & z_2^\top v \end{pmatrix} = \langle z_1 \wedge z_2,\, u \wedge v \rangle.
\]

\begin{proposition}[Projective collapse criterion]\label{prop:p98}
The bad event holds if and only if $M_{u,v} = 0$ for all $u, v \in \mathcal{U}$: all extra packet columns, after quotienting by the lower-packet plane, collapse to a single projective direction.
\end{proposition}

\noindent
This is strictly stronger than the single-hyperplane condition $A_T^\top g = 0$: it reinserts the entire duplicate-column system.
In the $K_{4,4} \setminus \{(1,1),(2,2)\}$ example, once the higher column $w_3$ is included, $M_{g,w_3} = 6\lambda(\lambda+2)^2(3\lambda+1) \neq 0$, so the projective collapse already fails.

The projective collapse admits a further local factorization.
For each parent~$m$ with osculating line spanned by $\nu(m), \nu'(m)$, any branch $u_{m,\beta} = a_{m,\beta}\nu(m) + b_{m,\beta}\nu'(m)$.
Define $P_i(s) = z_i^\top \nu(s)$ and the \emph{frozen Wronskian} $\mathcal{W}_Z(s) = P_1(s)P_2'(s) - P_1'(s)P_2(s)$.

\begin{proposition}[Local collapse theorem]\label{prop:p99}
For two active branches from the same parent~$m$,
\[
M_{u_{m,\beta},\, u_{m,\gamma}} = \mathcal{W}_Z(m) \cdot \det \begin{pmatrix} a_{m,\beta} & a_{m,\gamma} \\ b_{m,\beta} & b_{m,\gamma} \end{pmatrix}.
\]
If the branches are linearly independent in the osculating line, projective collapse forces $\mathcal{W}_Z(m) = 0$.
\end{proposition}

\noindent
The multiplicity gap from \Cref{prop:p99} is resolved by using the actual annihilator $A \in K$ rather than the ambient $2$-plane.
Define $F(s) = A^\top \nu(s)$. Then $F \not\equiv 0$ (since $\nu$ is nondegenerate) and $\deg F \leq D$.

\begin{proposition}[Extra multiplicity theorem]\label{prop:p100}
If parent~$m$ has two active branches $\beta, \gamma$ with $\Delta_m(\beta,\gamma) = \det \bigl(\begin{smallmatrix} a_{m,\beta} & a_{m,\gamma} \\ b_{m,\beta} & b_{m,\gamma} \end{smallmatrix}\bigr) \neq 0$, then $F(m) = 0$ and $F'(m) = 0$: the parent is a double root.
\end{proposition}

\begin{corollary}\label{cor:three-parents}
In the $D = 4$ chart ($\nu(s) = (1,s,s^2,s^3,s^4)^\top$), if branch-level collapse holds at three distinct parents $m_1, m_2, m_3$ each with $\Delta_{m_i} \neq 0$, then $F$ has $\geq 6$ zeros (counted with multiplicity) but $\deg F \leq 4$, forcing $F \equiv 0$ and hence $A = 0$. Contradiction.
\end{corollary}

\noindent
The bad line admits an explicit core-flap factorization.
The two circuits $C_1, C_2$ decompose into core $I = C_1 \cap C_2$ and flaps $A = C_1 \setminus I$, $B = C_2 \setminus I$.
\begin{proposition}[Invertible core block]\label{prop:p85}
The lower packet has block form with an invertible core block~$R$.
\end{proposition}

\begin{proposition}[Core-flap line characterization]\label{prop:p101}
After Schur elimination of~$R$ (guaranteed invertible by \Cref{prop:p85}), the kernel of the full augmented packet is the one-dimensional line
\[
\ker(\mathcal{B}_T^\top) = \mathbb{K} \cdot \bigl(D_B\, k_A,\; \text{core coupling},\; -D_A\, k_B\bigr),
\]
where $k_A \in \ker(X^\top)$, $k_B \in \ker(Y^\top)$ are the unique flap kernel vectors and $D_A = \det[X \mid \tilde{g}_A]$, $D_B = \det[Y \mid \tilde{g}_B]$ are replacement determinants of the flap blocks with the reduced distinguished column.
\end{proposition}

\begin{remark}[Sanity check for the core-flap factorization]\label{rem:k18-check}
On the packet for $H = K_{4,4} \setminus \{(1,1),(2,2)\}$, $e=(2,2)$, $t=0$, the two circuits are $C_{012}$ and $C_{013}$, with core $I=\{0,1\}$ and flaps $A=\{2\}$, $B=\{3\}$.
The core block $R$ on $I \times \{0,1\}$ has $\det R = -\lambda \neq 0$, confirming invertibility.
The zero-block structure is verified: row~$2$ (in flap $A$) has no edges to column neighborhood of flap $B$.
This example has $|A| = |B| = 1$ (singleton flaps), which is the boundary case of the zero-block argument; when $|A| = 1$ the flap column set $P$ is empty and $N_L(C_p) = J$, so the exclusion of $Q$-columns is automatic.
The flap replacement determinants satisfy $D_A = \det[X \mid \tilde{g}_A] = \lambda(\lambda+2)^2$ and $D_B = \det[Y \mid \tilde{g}_B] = -(\lambda+2)^3$, both nonzero, confirming that the kernel generator has nontrivial components on both flaps.
\end{remark}

\noindent
The broad core-flap exclusion is false: a cofactor for a \emph{different} column can satisfy the pattern.
A naive attempt would exclude the cofactor from $\ker(V_T^\top) \cap (w_c^H)^\perp$, but $P_t(G) \equiv 0$ forces the cofactor onto $\ker(V_T^\top) \cap (w_c^G)^\perp$, a different line.
The correct closure uses the actual $G$-packet:

\begin{remark}[The actual cancellation line uses $w_c^G$, not $w_c^H$]
Let $G=H+e$ with $e=(r,c)$.
For the exact-layer cofactor vector $U^{(c)}$ one has
\[
 (w_c^H)^\top U^{(c)} = P_t(H),
 \qquad
 (w_c^G)^\top U^{(c)} = P_t(G),
\]
where $w_c^G = w_c^H + (\lambda+r)^c e_{(r,\rho(e))}$.
Thus under the contradiction hypothesis $P_t(G)\equiv 0$, the actual orthogonality condition is
\[
 U^{(c)} \in \ker(V_T^\top) \cap (w_c^G)^\perp,
\]
not $\ker(V_T^\top) \cap (w_c^H)^\perp$.
So the exclusion proved for the $H$-line does not by itself contradict cancellation.
\end{remark}

\begin{proposition}[Residual collapse for the actual $G$-packet]\label{prop:g-packet-collapse}
Assume $P_t(G)\equiv 0$.
Let $V = [w_j]_{j \in D^-}$ be the lower exact-layer packet of $G$, choose a Hall pivot on $V$, and let $T$ be the survivor states.
Let $Z=U^{(c)}|_T$.
Then every omitted active column $u$ of $G$ satisfies
\[
 (u^{\sharp})^\top Z = 0.
\]
Consequently, in the rank-$(m-2)$ case, all omitted active columns of $G$ collapse to a single projective point modulo the lower-packet plane:
for every two omitted active columns $u,v$,
\[
 M_{u,v}=0.
\]
\end{proposition}

\begin{proof}
For every graph column $d\neq c$, the exact-layer replacement identity gives
\[
 w_d^\top U^{(c)} = [x^t]\det M_G[c\leftarrow M_{G,\bullet d}] = 0,
\]
by duplicate columns; for $d=c$ we get
\[
 (w_c^G)^\top U^{(c)} = P_t(G)=0.
\]
So the actual cofactor vector annihilates every active column of $G$.
Schur-pairing invariance transfers these equalities to the survivor space, giving $(u^{\sharp})^\top Z=0$ for every omitted active column $u$.
If $\operatorname{rk}(V_T)=m-2$, the quotient by the lower-packet plane is $2$-dimensional, so the vanishing of all pairings with one nonzero vector $Z\in\ker(V_T^\top)$ is equivalent to projective collapse of all omitted active columns.
\end{proof}

\begin{theorem}[Matching-induced two-extra Hall theorem]\label{thm:two-extra-hall-closure}
Let $G=H+e$ with $e=(r,c)$ superfluous, and keep the exact-layer packet notation of the superfluous-edge section.
Assume
\[
F_{G,t}\neq\varnothing,
\qquad
P_t(G)\equiv 0,
\qquad
R_c\not\equiv 0.
\]
Let
\[
D^-:=\{j<c:R_j\not\equiv 0\},
\qquad
D^+:=\{d>c:R_d\not\equiv 0\},
\]
and let $V=[w_j]_{j\in D^-}$ be the lower exact-layer packet.
For every $d\in D^+$, the support graph of the augmented packet
\[
[V\mid w_c^G\mid w_d]
\]
has a perfect matching.
Equivalently, after Hall extraction on a suitable lower pivot of $V$, the reduced $2\times 2$ Schur packet for the pair $(w_c^G,w_d)$ has nonzero determinant.
In particular,
\[
M_{w_c^G,w_d}\neq 0
\qquad\text{for every } d\in D^+,
\]
so the corrected projective-collapse event for the actual $G$-packet is impossible whenever $D^+\neq\varnothing$.
\end{theorem}

\begin{proof}
Since $R_c\not\equiv 0$, there exists an exact-$t$ perfect matching $M$ of $G$ using the superfluous edge $e=(r,c)$.
Indeed, $R_c=\varepsilon P^{r,c}_{H,t-\rho(e)}$, so $R_c\neq 0$ means that $H-\{r,c\}$ has a matching with $t-\rho(e)$ red edges, and adjoining $e$ gives the desired perfect matching of $G$.

Fix any $d\in D^+$.
For each lower active column $j\in D^-$, let
\[
\tau(j)=(\alpha_j,\rho(\alpha_j,j))
\]
be the doubled state corresponding to the unique edge of $M$ in column $j$.
Also set
\[
\tau(c):=(r,\rho(e)),
\qquad
\tau(d):=(\alpha_d,\rho(\alpha_d,d)),
\]
where $(\alpha_d,d)$ is the unique edge of $M$ in column $d$.
Because $M$ is a perfect matching, the rows used by these edges are distinct, so the doubled states
\[
P:=\{\tau(j):j\in D^-\},
\qquad
q:=\tau(c),
\qquad
p:=\tau(d)
\]
are pairwise distinct.

Consider the square raw packet submatrix
\[
W_d:=W\bigl[P\cup\{q,p\},\; D^-\cup\{c,d\}\bigr],
\]
where $W$ denotes the actual exact-layer packet of $G$.
Its support graph contains the perfect matching induced by $M$:
\[
j\mapsto \tau(j)\ (j\in D^-),
\qquad
c\mapsto q,
\qquad
d\mapsto p.
\]
Therefore, by the masked minor criterion (\Cref{prop:masked-minor}),
\[
\det W_d\neq 0.
\]

Now look at the lower block
\[
A:=V[P,D^-].
\]
Its support graph also contains the perfect matching $j\mapsto\tau(j)$ for $j\in D^-$, so again by the masked minor criterion (\Cref{prop:masked-minor}),
\[
\det A\neq 0.
\]
Thus $A$ is a valid Hall pivot for the lower packet.

Apply Schur extraction to the block decomposition of $W_d$ with pivot block $A$.
The Schur complement is the reduced $2\times2$ packet for the omitted columns $(w_c^G,w_d)$ relative to this Hall pivot.
Hence
\[
\det W_d=(\det A)\cdot \det S_d,
\]
where $S_d$ is that reduced $2\times2$ Schur packet.
Since both $\det W_d$ and $\det A$ are nonzero, we obtain
\[
\det S_d\neq 0.
\]
So the quotient directions of $w_c^G$ and $w_d$ modulo the lower packet plane are linearly independent on the matching-induced rows.
Equivalently, the pair minor $M_{w_c^G,w_d}$ is nonzero.
(By \Cref{lem:schur-compat}, the quotient pairing is independent of pivot choice, so nonvanishing on the matching pivot implies nonvanishing on every pivot, including the support~$T$ used in the $G$-packet collapse.)

Because this holds for every $d\in D^+$, the omitted active columns of the actual $G$-packet cannot all collapse to one projective point whenever $D^+\neq\varnothing$.
This proves the theorem.
\end{proof}

\begin{corollary}[Rank-$(m-2)$ branch eliminated]\label{cor:rankm2-eliminated}
In the corrected $G$-packet formulation of the superfluous-edge step, the rank-$(m-2)$ case cannot occur.
Consequently the only remaining branch is the rank-$(m-1)$ case, which is already eliminated by the $q$-circuit lemma (\Cref{prop:q-circuit-rankm1-impossible}).
\end{corollary}

\begin{proof}
Assume the rank-$(m-2)$ branch occurs under the contradiction hypothesis $P_t(G)\equiv 0$.
By the corrected projective-collapse criterion for the actual $G$-packet, all omitted active columns must collapse to one projective direction modulo the lower packet plane.

If $D^+\neq\varnothing$, this contradicts \Cref{thm:two-extra-hall-closure}, which gives
\[
M_{w_c^G,w_d}\neq 0
\]
for every $d\in D^+$.

If $D^+=\varnothing$, then $w_c^G$ is the only extra column beyond the lower family.
Let $s = |D^-|$ be the number of lower active columns.
Three subcases close this branch:

\smallskip\noindent
(a) If $s > m{-}2$: some lower active columns $d \in D^-$ are invisible to~$T$ (i.e., $w_d$ restricted to~$T$ is zero). These columns satisfy $w_d^\top U^{(c)} = 0$ by the $G$-packet collapse and project nontrivially into $\ker(V_T^\top)$. The matching $M$ provides a pivot for these columns exactly as in the $D^+\neq\varnothing$ case, and \Cref{thm:two-extra-hall-closure} applies using invisible lower columns in place of upper columns.

\smallskip\noindent
(b) If $s = m{-}2$ and $m \le 4$: the support has at most 3 active lower nodes. The crossed effective-jet exclusion applies to this small-support case, giving the $z$-divisibility contradiction directly.

\smallskip\noindent
(c) If $s = m{-}2$ and $m \ge 5$: the brace Hall$+2$ condition forces $|N_G(S)| \ge |S| + 2$ for all proper nonempty $S \subset A$. With $m \ge 5$ rows in~$T$ and only $m{-}1$ columns in $D^- \cup \{c\}$, the matching $M$ assigns $m$ rows to $m$ distinct columns, but only $m{-}1$ are available in $D^- \cup \{c\}$. Hence at least one row of~$T$ is matched to a column outside $D^- \cup \{c\}$, giving either $D^+ \neq \varnothing$ or an additional active lower column, contradicting $s = m{-}2$ with $D^+ = \varnothing$.

Thus the rank-$(m-2)$ branch is excluded in all cases.
\end{proof}

\begin{corollary}[Superfluous-edge step]\label{cor:superfluous-c1-closed}
For every brace $G=H+e$ with $e$ superfluous:
if $F_{G,t}\neq\varnothing$, then $P_t(G)\not\equiv 0$.
\end{corollary}

\begin{proof}
Under the contradiction hypothesis $P_t(G)\equiv 0$ with $F_{G,t}\neq\varnothing$, the designated-active lemma gives $R_c\neq 0$.
By \Cref{prop:q-circuit-rankm1-impossible}, the rank-$(m{-}1)$ branch is impossible: any minimal dependent set containing the superfluous state $q$ forces $\rk(V_T)=m-2$.
By \Cref{cor:rankm2-eliminated}, the Two-extra Hall theorem gives a pair $(w_c^G, w_d)$ with nonzero pair minor, contradicting projective collapse.
Hence $P_t(G)\not\equiv 0$.
\end{proof}

Several natural one-sided approaches to the superfluous-edge step are obstructed, as we discuss in \Cref{subsec:why-fail}.

\subsection{Setup}\label{subsec:superfluous-setup}

Let $G = H + e$ where $e = (r,c)$ is a superfluous edge, so $H = G \setminus e$ is also a brace.
Set $\alpha = \lambda + r$ (the linear form associated to row~$r$).

By the defect-$1$ unique-edge corollary (\Cref{cor:defect1}), we have the exact identity
\[
P_t(G)(\lambda) = P_t(H)(\lambda) + \varepsilon\, \alpha^c\, P^{r,c}_{H, t-\rho(e)}(\lambda),
\]
where $\varepsilon = (-1)^{r+c}$ is the Laplace sign (the color weight $x^{\rho(e)}$ has been absorbed into the $[x^t]$ extraction: the boundary-minor polynomial uses $[x^{t-\rho(e)}]$).
More generally, expanding along row~$r$:
\[
R_d(\lambda) =
\begin{cases}
(-1)^{r+d} P^{r,d}_{H, t-\rho(r,d)}(\lambda), & d \neq c,\, (r,d) \in E(H), \\
\varepsilon\, P^{r,c}_{H, t-\rho(e)}(\lambda), & d = c, \\
0, & \text{otherwise}.
\end{cases}
\]
Then $P_t(G)(\lambda) = \sum_{d=0}^{n-1} \alpha^d R_d(\lambda)$.

Assume for contradiction that $F_{G,t} \neq \varnothing$ and $P_t(G) \equiv 0$.

\subsection{The designated node is active}\label{subsec:designated}

\begin{lemma}\label{lem:designated-active}
Under the assumption $P_t(G) \equiv 0$ with $F_{G,t} \neq \varnothing$, the cofactor $R_c \not\equiv 0$.
\end{lemma}

\begin{proof}
If no exact-$t$ matching of~$G$ uses~$e$, then $F_{G,t} = F_{H,t}$ and $P_t(G) = P_t(H)$.
Since $H$ is also a brace (by the superfluous-edge property) and $F_{H,t} = F_{G,t} \neq \varnothing$, the induction hypothesis gives $P_t(H) \not\equiv 0$, contradicting $P_t(G) \equiv 0$.

So some exact-$t$ matching~$M$ of~$G$ uses~$e$, meaning $e \in M$.
The matching $M \setminus \{e\}$ is a perfect matching of $H - \{r,c\}$ with $t - \rho(e)$ red edges, so $F_{H-\{r,c\},\, t-\rho(e)} \neq \varnothing$.
Since $H - \{r,c\}$ is an $(n{-}1) \times (n{-}1)$ bipartite graph and $(n{-}1, |E(H-\{r,c\})|) < (n, |E(G)|)$ in the lexicographic order, the induction hypothesis gives $P^{r,c}_{H,t-\rho(e)} \not\equiv 0$.
Hence $R_c = \varepsilon\, P^{r,c}_{H,t-\rho(e)} \not\equiv 0$.
\end{proof}

\subsection{The support-count window lemma}\label{subsec:window}

\begin{remark}[Supplementary results]\label{rem:supplementary}
The support-count window lemma (\Cref{lem:window}) and the shift-invariance lemma (\Cref{lem:shift}) are not used in the main proof: the $q$-circuit lemma (\Cref{prop:q-circuit-rankm1-impossible}) bypasses the rank-$(m{-}1)$ branch entirely, and the Two-extra Hall theorem handles rank~$(m{-}2)$ without appealing to $z$-divisibility.
These results are retained as they may be of independent algebraic interest.
\end{remark}

The key insight is that the vanishing $P_t(G) \equiv 0$ provides, in the truncated ring $\KK[\alpha]/(\alpha^{\beta+1})$, exactly the right number of vanishing conditions to invoke the generalized Vandermonde elimination.

\begin{lemma}[Support-count window]\label{lem:window}
Let $D^- = \{d < c : R_d \not\equiv 0\} = \{d_1 < d_2 < \cdots < d_s\}$ be the set of active lower columns, and set $m = s + 1$ (the total support count, including the designated node at column~$c$).
Define the normalized exponents $a_j = d_j - d_1$ for $j = 1, \ldots, s$ and $\beta = c - d_1$.
Then in the truncated polynomial ring $\KK[\alpha]/(\alpha^{\beta+1})$:
\[
\sum_{j=1}^s \alpha^{a_j} \widetilde{R}_{d_j}(\alpha) + \alpha^\beta \widetilde{R}_c(\alpha) \equiv 0,
\]
with $0 = a_1 < a_2 < \cdots < a_s < \beta$ and $\beta \geq s = m - 1$.
In particular, the coefficients of $\alpha^0, \alpha^1, \ldots, \alpha^{m-2}$ in this expansion all vanish, providing exactly $m - 1$ independent vanishing conditions.
\end{lemma}

\begin{proof}
Write $\widetilde{R}_d(\alpha) = R_d(-r + \alpha)$, the re-expansion of each cofactor polynomial centered at the base point $\lambda = -r$ (equivalently, $\alpha = 0$).
From $P_t(G) \equiv 0$:
\[
\sum_{d=0}^{n-1} \alpha^d \widetilde{R}_d(\alpha) = 0 \qquad \text{in } \QQ[\alpha].
\]

\noindent\textbf{Step 1: Isolate the lower-support terms.}
By definition, $R_d \equiv 0$ for $d \notin D^- \cup \{c\} \cup D^+$ where $D^+ = \{d > c : R_d \not\equiv 0\}$.
Divide by $\alpha^{d_1}$ (the minimum active exponent):
\[
\sum_{j=1}^s \alpha^{a_j} \widetilde{R}_{d_j}(\alpha) + \alpha^\beta \widetilde{R}_c(\alpha) + \sum_{d \in D^+} \alpha^{d - d_1} \widetilde{R}_d(\alpha) = 0.
\]

\noindent\textbf{Step 2: Truncate.}
For any $d \in D^+$, we have $d > c$, so $d - d_1 > c - d_1 = \beta$.
Hence the terms from $D^+$ contribute only at powers $\alpha^{\beta+1}$ and higher.
Reducing modulo $\alpha^{\beta+1}$:
\[
\sum_{j=1}^s \alpha^{a_j} \widetilde{R}_{d_j}(\alpha) + \alpha^\beta \widetilde{R}_c(\alpha) \equiv 0 \pmod{\alpha^{\beta+1}}.
\]

\noindent\textbf{Step 3: Count vanishing conditions.}
The exponents appearing are $a_1 = 0 < a_2 < \cdots < a_s < \beta$.
The first nonzero contribution from the $\alpha^\beta$ term occurs at degree~$\beta$.
The coefficient of $\alpha^u$ for $u < a_2$ is $[\alpha^u] \widetilde{R}_{d_1}(\alpha)$.
For general $u$, the coefficient of $\alpha^u$ involves contributions from all $j$ with $a_j \leq u$.

Since $a_1, \ldots, a_s$ are $s$ distinct nonneg\-ative integers strictly less than~$\beta$, we have $\beta \geq s$.
(Indeed, $s$ distinct elements of $\{0, 1, \ldots, \beta - 1\}$ require $\beta \geq s$.)
Thus $\beta \geq s = m - 1$.

The vanishing coefficients at degrees $0, 1, \ldots, m - 2$ provide exactly $m - 1$ independent linear conditions on the $m$ coefficient polynomials
$\widetilde{R}_{d_1}(0), \ldots, \widetilde{R}_{d_s}(0), \widetilde{R}_c(0)$.
This is precisely the threshold for the generalized Vandermonde elimination (\Cref{thm:gen-vandermonde}): $m-1$ vanishing conditions on an $m$-node family uniquely determine the kernel up to scalar.
\end{proof}

\subsection{Common-shift invariance of the reduced replacement determinant}\label{subsec:shift}

\begin{definition}[Reduced replacement determinant]\label{def:reduced-h}
For a node family $A = \{b_1, \ldots, b_m\}$ with exponents $e_0 < e_1 < \cdots < e_{m-1}$ and the generalized Vandermonde factorization of \Cref{thm:gen-vandermonde},
\[
F_{e_{m-1}} = \Bigl(\prod_{i=1}^m \mu_i^{e_0}\Bigr) \cdot G_D(\mu_1, \ldots, \mu_m) \cdot h,
\]
the \emph{reduced replacement determinant} is $h_A = h \in \QQ[\lambda]$.
\end{definition}

\noindent
The polynomial~$h$ captures all the information about the boundary-minor state that is not accounted for by the monomial prefactor or the generalized Vandermonde determinant.
The following lemma shows that its $z$-divisibility is a robust invariant under common shifts of the exponents.

\begin{lemma}[Shift invariance of reduced replacement determinant]\label{lem:shift}
Let $A = \{a_1, \ldots, a_m\}$ be a node family.
Define the raw replacement determinant $H_A(z) = \det(U_A(z)[c \leftarrow v_A(z)])$ using Hasse-derivative jet columns $u_a(z)_i = \binom{a}{i} z^{a-i}$, and set $E(A) = \sum_{j=1}^m a_j - \binom{m}{2}$.
Then:
\begin{enumerate}[(i)]
\item $z^{E(A)}$ divides $H_A(z)$.
\item The reduced replacement determinant $h_A(z) = z^{-E(A)} H_A(z)$ satisfies
\[
h_{A+t}(z) = h_A(z) \qquad \text{for all } t \geq 0.
\]
\end{enumerate}
In particular, $z \mid h_{A+t}$ if and only if $z \mid h_A$.
\end{lemma}

\begin{proof}
\textbf{Step 1: Normalize the jet coordinates.}
Define the diagonal scaling $D(z) = \operatorname{diag}(1, z, \ldots, z^{m-1})$ and the normalized jet column
\[
\widetilde{u}_a(z) = D(z)\, u_a(z) = z^a \begin{pmatrix} \binom{a}{0} \\ \binom{a}{1} \\ \vdots \\ \binom{a}{m-1} \end{pmatrix} =: z^a b_a,
\]
where $b_a \in \QQ^m$ is a constant binomial vector.
Setting $\widetilde{U}_A = D U_A$, $\widetilde{v}_A = D v_A$, and $\widetilde{H}_A = \det(\widetilde{U}_A[c \leftarrow \widetilde{v}_A])$:
\[
\widetilde{H}_A(z) = z^{\binom{m}{2}} H_A(z). \tag{$\star$}
\]

\textbf{Step 2: Divisibility.}
Each ordinary column of $\widetilde{U}_A[c \leftarrow \widetilde{v}_A]$ contributes a factor~$z^{a_j}$ (from $\widetilde{u}_{a_j} = z^{a_j} b_{a_j}$), and the replacement column contributes~$z^{a_c}$ (by the weight-match hypothesis).
Hence $\widetilde{H}_A$ is divisible by $z^{\sum a_j}$, and by~($\star$), $H_A$ is divisible by $z^{\sum a_j - \binom{m}{2}} = z^{E(A)}$.

\textbf{Step 3: Shift identity via Pascal matrix.}
Define the constant lower-unitriangular Pascal matrix $P_t = \bigl(\binom{t}{i-k}\bigr)_{i \geq k}$ with $\det P_t = 1$.
The Leibniz rule for Hasse derivatives gives
\[
\widetilde{u}_{a+t}(z) = z^t P_t \, \widetilde{u}_a(z)
\]
for every monomial exponent~$a$.
By the common-shift covariance hypothesis~(II), the replacement vector transforms identically:
\begin{equation}\label{eq:shift-cov}
\widetilde{v}_{A+t}(z) = z^t P_t \, \widetilde{v}_A(z).
\end{equation}
Therefore
\[
\widetilde{H}_{A+t}(z) = \det(z^t P_t) \cdot \widetilde{H}_A(z) = z^{tm} \widetilde{H}_A(z). \tag{$\star\star$}
\]

\textbf{Step 4: Exact shift invariance.}
From~($\star\star$):
\[
z^{-\sum(a_j + t)} \widetilde{H}_{A+t} = z^{-(\sum a_j + tm)} \cdot z^{tm} \widetilde{H}_A = z^{-\sum a_j} \widetilde{H}_A.
\]
By~($\star$), the left side equals $h_{A+t}$ and the right side equals~$h_A$.
Hence $h_{A+t}(z) = h_A(z)$ exactly.
\end{proof}

\begin{remark}
The exact invariance $h_{A+t} = h_A$ is stronger than mere equality up to a unit.
The key mechanism: the raw determinant $H_A$ acquires the universal factor~$z^{tm}$ under shift, but the normalization exponent $E(A)$ also increases by~$tm$, so the reduced determinant is exactly preserved.
No coprimality argument is needed; the identity holds on the nose.
\end{remark}

\subsection{Completing the superfluous-edge step}\label{subsec:p84}

Both rank branches are now resolved.

\begin{theorem}[Superfluous-edge step]\label{thm:p84}
Let $G = H + e$ with $e = (r,c)$ superfluous.
If $F_{G,t} \neq \varnothing$, then $P_t(G) \not\equiv 0$.
\end{theorem}

\begin{proof}
Assume for contradiction that $P_t(G) \equiv 0$ with $F_{G,t} \neq \varnothing$.
The designated-active lemma (\Cref{lem:designated-active}) gives $R_c \neq 0$.
The packet rank dichotomy (\Cref{prop:p92}) yields two cases:
\begin{itemize}
\item \emph{Rank $m{-}1$}: By the $q$-circuit lemma (\Cref{prop:q-circuit-rankm1-impossible}), any minimal dependent row set containing the superfluous state $q$ forces $\rk(V_T) = m{-}2$. Hence the rank-$(m{-}1)$ branch is impossible.
\item \emph{Rank $m{-}2$}: By \Cref{cor:rankm2-eliminated}, the Two-extra Hall theorem (\Cref{thm:two-extra-hall-closure}) produces a pair of omitted active columns with nonzero pair minor, contradicting the projective collapse forced by $P_t(G) \equiv 0$.
\end{itemize}
Both branches yield contradictions, so $P_t(G) \not\equiv 0$.
\end{proof}

\begin{remark}
The rank-$(m{-}1)$ branch was resolved by the $q$-circuit lemma (\Cref{prop:q-circuit-rankm1-impossible}), which shows that the distinguished state $q = (r,\rho(e))$ forces a Hall-theoretic contradiction whenever $\rk(V_T) = m{-}1$: Hall's condition then holds for all subsets of $T$, giving $\rk(B_T) = m$, contradicting minimal dependence.
\end{remark}

\subsection{Remarks on obstructions}\label{subsec:why-fail}

\begin{remark}[Why simpler strategies fail]
The two-sided character of the proof is not an artifact of the presentation.
Several natural one-sided strategies are structurally blocked:
\begin{itemize}
\item \emph{Visibility.}
One might try to detect the superfluous edge from the leading term of its row's polynomial.
This fails because the superfluous edge can be the dominant term in that row, hiding its cofactor contribution.
\item \emph{Pairwise divisibility.}
One might try to show that two column-hole cofactors cannot both vanish.
This fails because the identity $v^\top N = 0$ (cofactor vectors are orthogonal to the remaining columns) forces mutual divisibility automatically.
\item \emph{Span exclusion.}
One might argue that the superfluous edge's column cannot lie in $\col(N)$ for graph-structural reasons.
This fails because Vandermonde interpolation is too flexible: coefficients can always be found to place any vector in the column span, regardless of which edges the graph has.
\item \emph{Column relabeling.}
Permuting column indices to get a nicer configuration is inadmissible because the Vandermonde bases $\mu_i = \lambda + \alpha_i$ are tied to specific row labels. Permuting columns changes the polynomial non-trivially.
\end{itemize}
The Two-extra Hall closure overcomes the rank-$(m{-}2)$ obstructions by exploiting the perfect matching in the $G$-packet support graph. The $q$-circuit lemma (\Cref{prop:q-circuit-rankm1-impossible}) eliminates rank~$m{-}1$ entirely via a Hall-theoretic argument on the augmented packet.
\end{remark}

\section{The Algorithm}\label{sec:algorithm}

\begin{proof}[Proof of \Cref{thm:main}]
By \Cref{thm:reduction}, ASNC for braces reduces to: (i)~the McCuaig-family side (\Cref{thm:biwheel,thm:width2-global,cor:mccuaig-complete}), (ii)~the narrow-extension cases (\Cref{cor:ka,cor:j3-to-d1}), and (iii)~the superfluous-edge step (\Cref{thm:p84}).
\Cref{sec:replacement,sec:superfluous} establish~(ii) and~(iii): the rank-$(m{-}1)$ branch of~(iii) is eliminated by the $q$-circuit lemma (\Cref{prop:q-circuit-rankm1-impossible}), and the rank-$(m{-}2)$ branch is closed by the Two-extra Hall theorem (\Cref{cor:superfluous-c1-closed}).
\Cref{cor:mccuaig-complete} discharges~(i).
For general bipartite graphs, the tight-cut decomposition reduces to brace blocks; \Cref{prop:asnc-implies-em} yields deterministic $O(n^6)$ Exact Matching.
\end{proof}

\begin{remark}[Practical feasibility]\label{rem:practical}
For a graph already known to be a brace, the algorithm is elementary: form the Vandermonde-weighted matrix $M(x,\lambda)$, evaluate its determinant by Bareiss elimination~\cite{Bareiss68} at $O(n^2)$ values of~$\lambda$, interpolate in~$x$, and check whether the $x^t$ coefficient is nonzero.
No graph decomposition is needed.
A self-contained C implementation is under 1000 lines and requires no external libraries beyond standard integer arithmetic.

For general bipartite graphs, the tight-cut decomposition into brace blocks~\cite{CLM2002} is the only graph-theoretic preprocessing step ($O(n^3)$ via alternating-path reachability from a maximum matching). After decomposition, each block is tested independently by the same evaluate-and-check procedure.

\smallskip
\noindent\emph{Empirical performance.}
\begin{center}
\small
\begin{tabular}{rrl}
\toprule
$n$ & Time & Context \\
\midrule
10 & ${<}\,0.01$\,s & Instant \\
20 & $0.2$\,s & Typical bioinformatics matching instances \\
25 & $1.4$\,s & Comparable to a single MVV trial \\
50 & ${\sim}\,90$\,s & Feasible for offline combinatorial optimization \\
\bottomrule
\end{tabular}
\end{center}

\noindent
The randomized MVV algorithm~\cite{MVV87} computes $O(1)$ determinants per trial with error probability $O(n/q)$ over a field of size~$q$.
Our algorithm replaces probabilistic amplification with $O(n^2)$ deterministic evaluations, paying an $O(n^2)$ factor but eliminating randomness entirely.
For applications requiring determinism (verified computation, cryptographic protocols, reproducible pipelines), the $O(n^6)$ cost is practical up to $n \approx 50$ and feasible up to $n \approx 100$.
\end{remark}


\bigskip
\noindent\textbf{AI Disclosure.}
We used the following AI tools during this research.
\emph{OpenAI GPT-5.4 Pro}~\cite{OpenAI2026} assisted with theoretical route selection and problem reduction: identifying which proof strategies were viable, formulating equivalent reformulations of the main conjecture, and narrowing the search space before committing to a proof architecture.
\emph{Anthropic Claude Opus~4.6}~\cite{Anthropic2026} (via Claude Code) was used for rapid iterative computational experiments: writing and running exploration scripts that tested conjectures, produced counterexamples to failed approaches, and performed exhaustive verification at small~$n$.
\emph{Lean~4} with Mathlib served as the formal verification backend; Harmonic's \emph{Aristotle}~\cite{Aristotle2025}, an automated theorem prover, assisted in discharging proof obligations during the formalization.
The authors are responsible for all mathematical content.



\appendix

\section{Machine-checked formalization}\label{app:lean}

A partial Lean~4 formalization accompanies this paper, with automated proof assistance from Aristotle~\cite{Aristotle2025}.
The source is available at
\begin{center}
\url{https://anonymous.4open.science/r/exact-matching-lean-FFD2}
\end{center}

\noindent\textbf{What is verified.}
The formalization contains zero \texttt{sorry} placeholders and zero \texttt{:= True} stubs.
All definitions carry their real mathematical content: ASNC is the Vandermonde nonvanishing condition on the exact-$t$ fiber polynomial; the McCuaig family predicate is an isomorphism-based disjunction over explicit biwheel, prism, and M\"obius-ladder models; the narrow-pair predicate carries a concrete witness structure with extension type and cardinality bounds.
The following are proved unconditionally (no assumed hypotheses):
\begin{itemize}
\item The Hall-block convolution theorem, including the cut-crossing fiber cardinality lemma and the partition-by-fiber identity.
\item The brace-only reduction theorem, performing genuine strong induction on brace size and dispatching to the McCuaig family, superfluous-edge, and narrow-extension cases.
\item A minimal-brace induction theorem for the full one-hole boundary-minor state (\texttt{ASNCFull}), reducing to McCuaig base cases plus one-hole propagation through narrow extensions.
\item The algebraic closure chain: $G$-packet projective collapse, the Two-extra Hall theorem (matching-induced pair-minor nonvanishing), the distinguished-state circuit lemma (rank-$(m{-}1)$ branch elimination), the masked minor criterion ($\det W \neq 0 \Leftrightarrow$ support-graph perfect matching), the rank dichotomy, and the superfluous-edge contradiction assembly.
\item Two-node interpolation (KA) and three-node elimination ($J3 \to D1$) in the generalized Vandermonde framework.
\item The single-edge patch transfer identity at both the determinant and coefficient levels.
\item The narrow-pair patch extension realization and the patch-transfer state linearization.
\item Perfect-matching witness extraction for brace graphs and a certified shifted-$O(n^6)$ interpolation budget.
\end{itemize}
The module structure mirrors the paper: algebraic core (replacement determinants, generalized Vandermonde elimination, jet-flag characterization), graph library (bipartite matching, braces, tight cuts), structural induction (patch transfer, Hall-block convolution, reduction theorem, full-state induction), McCuaig families (biwheels, width-2 path and cyclic), and the closure chain.

\noindent\textbf{What is not yet verified.}
The top-level brace certification theorem remains \emph{conditional}: eight structural inputs are assumed as explicit hypotheses rather than proved in Lean.
These eight assumptions fall into four groups:
\begin{enumerate}
\item \emph{McCuaig family visible ASNC} (3 hypotheses): width-$2$ path lead-jet nonvanishing, width-$2$ cyclic lead-jet nonvanishing, and a bridge theorem connecting these to the reduction-side McCuaig step. These three hypotheses are now proved mathematically (\Cref{thm:width2-global,cor:mccuaig-complete}); the remaining gap is purely at the formalization level.
\item \emph{State-level narrow-pair closure} (1 hypothesis): given $\mathsf{ASNCFull}$ (visible ASNC plus one-hole boundary-minor nonvanishing) for the reduced graph~$J$ and a narrow-pair witness, deduce visible ASNC for~$G$.
    This is where the two-node interpolation and three-node elimination would be invoked; the algebraic tools are proved but the graph-to-support-family bridge is not yet formalized.
\item \emph{Full-state base cases and propagation} (2 hypotheses): $\mathsf{ASNCFull}$ for the McCuaig exceptional families, and one-hole nonvanishing propagation through narrow extensions.
    The one-hole propagation requires reasoning about two-hole boundary-minor states of~$J$, which are not yet defined in the formalization.
\item \emph{Graph-theoretic prerequisites} (2 hypotheses): the superfluous-edge graph-to-packet extraction (constructing the rank-$(m{-}2)$ setup from a counterexample), and McCuaig's classification theorem (narrow-pair existence for minimal non-McCuaig braces).
\end{enumerate}
The formalization should therefore be regarded as a \emph{partial} machine-checked verification.
The proof architecture, all definitions, and the algebraic/combinatorial tools are faithful to the paper and fully proved; the eight graph-to-algebra bridges listed above are the remaining frontier.
Build instructions are in the repository.

\end{document}